%% file: main.tex
\documentclass[times,11pt]{elsarticle}
\usepackage{fullpage}
\usepackage[OT1]{fontenc}
\usepackage{amsmath,amssymb,amsthm,subfig,bm,units,tikz,booktabs}
\usetikzlibrary{arrows}
\usepackage{url}
\usepackage[ruled,algo2e]{algorithm2e}
\usepackage[mathscr]{euscript}
\usepackage{mdframed,booktabs,rotating}
\usepackage[suffix=]{epstopdf}
\DeclareGraphicsExtensions{.png,.pdf}
\usepackage[title,toc,titletoc,page]{appendix}
\graphicspath{{pic/}}

\newcommand{\dd}{\mathrm d}

\theoremstyle{definition}
\newtheorem{theorem}{Theorem}

\newtheorem{corollary}{Corollary}
\newtheorem{remark}{Remark}

\usepackage{algorithm, algorithmic}

\DeclareUnicodeCharacter{2212}{-}

\begin{document}

\title{{\bf An Efficient Monte Carlo Simulation for Radiation Transport Based on Global Optimal Reference Field}}

\author[sjtu]{Minsheng Huang}
\ead{mingo.stemon@sjtu.edu.cn}
\author[pku]{Ruo Li\corref{cor}}
\ead{rli@math.pku.edu.cn}
\author[pku,nint]{Kai Yan}
\ead{1501110028@pku.edu.cn}
\author[nint]{Chengbao Yao}
\ead{yaocheng@pku.edu.cn}
\author[sjtu,ins]{Wenjun Ying}
\ead{wying@sjtu.edu.cn}

\cortext[cor]{Corresponding author}
\address[sjtu]{School of Mathematical Sciences, Shanghai Jiao Tong University, Shanghai, P.R.China.}
\address[pku]{School of Mathematical Sciences, Peking University, Beijing, P.R.China.}
\address[nint]{Northwest Institute of Nuclear Technology, Xi'an, P.R.China.} 
\address[ins]{MOE-LSC and Institute of Natural Sciences, Shanghai Jiao Tong University, Shanghai, P.R.China.}

\begin{abstract}
The reference field method, known as the difference formulation, is a key variance reduction technique for Monte Carlo simulations of thermal radiation transport problems. When the material temperature is relatively high and the spatial temperature gradient is moderate, this method demonstrates significant advantages in reducing variance compared to classical Monte Carlo methods. However, in problems with larger temperature gradients, this method has not only been found ineffective at reducing statistical noise, but in some cases, it even increases noise compared to classical Monte Carlo methods. The global optimal reference field method, a recently proposed variance reduction technique, effectively reduces the average energy weight of Monte Carlo particles, thereby decreasing variance. Its effectiveness has been validated both theoretically and numerically, demonstrating a significant reduction in statistical errors for problems with large temperature gradients. In our previous work, instead of computing the exact global optimal reference field, we developed an approximate, physically motivated method to find a relatively better reference field using a selection scheme. In this work, we reformulate the problem of determining the global optimal reference field as a linear programming problem and solve it exactly. To further enhance computational efficiency, we use the MindOpt solver, which leverages graph neural network methods. Numerical experiments demonstrate that the MindOpt solver not only solves linear programming problems accurately but also significantly outperforms the Simplex and interior-point methods in terms of computational efficiency. The global optimal reference field method combined with the MindOpt solver not only improves computational efficiency but also substantially reduces statistical errors. 

\end{abstract}


\maketitle

\input{grf_intro.tex}
\input{grf_model.tex}

\input{grf_method.tex}
\input{grf_scheme.tex}
\input{grf_result.tex}
\input{grf_conclusion.tex}

\bibliographystyle{unsrt}
\biboptions{numbers,sort&compress} 
\bibliography{ref}
\end{document}

%% file: grf_intro.tex
\section{Introduction}
Statistical error has long been a key challenge limiting the large-scale application of Monte Carlo methods in thermal radiative transfer \cite{howell1998monte,wollaber2016four}. Since the statistical error in Monte Carlo methods is inversely proportional to the square root of the number of particles, simply increasing the number of particles to improve computational accuracy is often impractical. To address this, various variance reduction techniques have been developed to mitigate statistical noise, including Russian roulette, splitting, and weight window methods, among others \cite{WOLLABER200808,noebauer2019monte,steinberg2022multi}. An important variance reduction technique in Monte Carlo methods is the reference field method, also referred to as the difference formulation \cite{mckinley2003comparison, BROOKS2005737}. In many thermal radiation transport problems, the radiation field closely approximates a blackbody radiation field within optically thick regions based on the current material temperature. In classical Monte Carlo simulations, changes in material energy during each time step are determined by statistically accounting for absorbed photon energy and deterministically calculating emitted photon energy in the region. Since the absorbed photon energy is derived statistically, the classical Monte Carlo method introduces statistical noise into the solution. The reference field method addresses this by introducing the concept of the difference field, or difference intensity, which is defined as the difference between the current radiation field and the blackbody radiation field at the current material temperature. The latter field serves as the reference field. 

When the reference field closely approximates the current radiation field, the ``energy" of the difference field is considered low relative to the current radiation field. Introducing the difference field eliminates the thermal emission term but introduces new source terms. In reference \cite{BROOKS2005737}, Brooks derived the thermal radiation transport equations using the difference formulation and developed a symbolic implicit Monte Carlo (IMC) method to solve these equations.
Early studies typically selected the blackbody radiation field as the reference field based on the current material temperature. In this paper, we refer to this formulation as the ``classical reference field formulation" for thermal radiative transfer and the associated variance reduction technique as the ``classical reference field method". The standard thermal radiation transport equations (equivalent to a zero reference field) are referred to as the ``standard formulation" for thermal radiation transport problems. While the reference field method has demonstrated its ability to reduce statistical errors significantly in some instances, research has also shown that in scenarios with significant temperature gradients, this method results in even larger statistical errors than before \cite{luu2010generalized, cleveland2010extension}. Consequently, this approach received limited attention. Later, Luu et al. \cite{luu2010generalized} found that when there was a significant discrepancy between the radiation field intensity and the blackbody radiation field intensity, the thermal emission term may vanish, but spatial and temporal derivative terms related to the blackbody radiation intensity still lead to elevated statistical noise. To address this issue, they proposed the generalized reference field method \cite{luu2010generalized}. However, the generalized reference field method does not provide a specific strategy for selecting the reference field. To solve this problem, we proposed the concept of the global optimal reference field method \cite{2021The}. In our previous work, rather than offering a direct numerical method for solving the global optimal reference field exactly, we developed a relatively effective selection scheme to find an acceptable and physically significant approximation for the reference field.

In this paper, we derive the generalized difference formulation of the IMC thermal radiation transport equations and provide an expression for the total energy of particles emitted at each time step. Following the original definition \cite{2021The}, we establish the global optimal reference field and reformulate its calculation as an optimization problem. A computational algorithm was developed to solve for the global optimal reference field, and it was proven that the global optimum must reside within a specific finite set. This optimization problem was successfully transformed into a linear programming problem, which was solved using the Simplex method. We further compared the computational efficiency of the global optimal reference field method against the classical IMC method for thermal radiation transfer problems. The numerical results demonstrated that the global optimal reference field method significantly reduces statistical noise. To enhance computational efficiency, we utilized the MindOpt solver \cite{mindopt}, which leverages Graph Neural Networks to solve the linear programming problem. Numerical experiments confirmed that this solver not only solves the problem more efficiently but also improves the overall performance of the global optimal reference field method in thermal radiation transport equations.


The rest of this paper is organized as follows. In Section 2, we brieﬂy reviewed the standard formulation and difference formulation of the thermal radiation transport equations. In Section 3, starting from the IMC thermal radiative transfer, we derived the generalized difference formulation of the IMC equations and then transformed the original problem of the global optimal reference field into an optimization problem. Section 4 provides an accurate computational method for the global optimal reference field. Section 5 presents relevant numerical examples, verifying the correctness and computational efficiency of our method. Finally, we make a brief conclusion in Section 6. 

%% file: grf_model.tex
\section{Mathematical model}\label{model}
Assuming satisfying the local thermodynamic equilibrium (LTE), the one-dimensional thermal radiation transport equations  neglecting the scattering phenomenon are given as:
\begin{subequations}
\begin{align}
& \frac{1}{c}\frac{\partial I}{\partial t}+\mu\frac{\partial I}{\partial x}+\sigma I=\sigma B+\frac{1}{4\pi} Q,\label{1a}\\
& \frac{\partial U_m}{\partial t}=2\pi\int_0^{\infty}\int_{-1}^1 \sigma(I-B)\dd\mu\dd\nu,\label{1b}
\end{align}
\label{1}
\end{subequations}
the initial conditions are:
\begin{align*}
I(x,0,\mu,\nu)& =I^i(x,\mu,\nu),\\
T(x,0)& =T^i(x),
\end{align*}
and the boundary conditions are:
\begin{align*}
I(0,t,\mu,\nu)& =I^l(t,\mu,\nu),~~ 0\le \mu \le 1,\\
I(X,t,\mu,\nu)& =I^r(t,\mu,\nu),~~ -1\le \mu \le 0,
\end{align*}
here, $I(x,t,\mu,\nu)$ is radiation intensity,  $x$ is the spatial variable,  $t$ is the time variable, the variable $\mu$ represents the cosine of the angle between the photon's direction of motion and the positive direction of the $x$-axis, the variable $\nu$ represents the frequency of the photon. $U_m(x,t)=C_v(x,t)T(x,t)$ is the material
energy density, where $C_v(x,t), T(x,t)$ denote the specific heat coefficient and the material temperature, respectively. 
$\sigma(x,\nu,T(x,t))$ denotes the absorption opacity and $Q(x,t,\nu)$ represents an independent external source. $B(\nu,T(x,t))$ is the Planck function, which is concretized as follows:
\begin{equation}
      B(\nu, T) = \frac{2h\nu^{2}}{c^{3}} \frac{1}{e^{h\nu/\kappa T} - 1} = \frac{ac}{4\pi}T^{4}b(\nu, T),
      \label{def:B}
\end{equation}
where $h, \kappa$ denote the Planck constant, and Boltzmann constant, respectively. $b(\nu, T)$ is the probability density function, which satisfies:
\begin{equation}
     \int_{0}^{\infty}  b(\nu,T) \dd \nu = 1.
    \label{def:sigp}
\end{equation}

\subsection{The classical reference field method}
According to the literature \cite{WOLLABER200808,wang2019quantitative}, the Monte Carlo method suffers from the inefficiency problem of solving the thermal radiation transport equations since the total energy of photons emitted by thermal sources at each time step can become exceedingly large as the temperature increases in optically thick regions.
In classical Monte Carlo methods, we require a large number of Monte Carlo particles to represent these thermal sources. However, most of these particles will be absorbed by the nearby medium, leading to plenty of ineffective particle histories being tracked and severely impacting both computational efficiency and accuracy. Once given that in optically thick regions, the radiation intensity $I$ is often approximately equal to the blackbody radiation intensity. A natural idea is to define a difference field intensity $D(x,t,\mu,\nu)$ as follows:
\begin{align}
D(x,t,\mu,\nu)=I(x,t,\mu,\nu)- B(\nu, T(x,t)).\label{2}
\end{align}
By substituting equation \eqref{2} into equations \eqref{1}, we can derive the set of thermal radiation transport equations described under the classical reference field formulation:
\begin{subequations}
\begin{align}
& \frac{1}{c}\frac{\partial D}{\partial t}+\mu\frac{\partial D}{\partial x}=\sigma D-\frac{1}{c}\frac{\partial B}{\partial t}-\mu\frac{\partial B}{\partial x}+\frac{1}{4\pi} Q,\\
& \frac{\partial U_m}{\partial t}=2\pi\int_0^{\infty}\int_{-1}^1 \sigma D\dd\mu\dd\nu.
\end{align}
\end{subequations}

The initial reference field intensity is taken as $B(\nu, T(x,0))$, and the intensity of the boundary reference fields are taken as $B^l(\nu,t)$ and $B^r(\nu,t)$, respectively. Thus, the initial conditions are:
\begin{subequations}
\begin{align}
D(x,0,\mu,\nu)& =I^i(x,\mu,\nu)-B(\nu,T(x,0)),\\
T(x,0)& =T^i(x),
\end{align}
\label{3}
\end{subequations}
the boundary conditions are:
\begin{subequations}
\begin{align}
D(0,t,\mu,\nu)& =I^l(t,\mu,\nu)-B^l(\nu,t),~~ 0\le \mu \le 1,\\
D(X,t,\mu,\nu)& =I^r(t,\mu,\nu)-B^r(\nu,t),~~ -1\le \mu \le 0.
\end{align}
\label{4}
\end{subequations}
Compared to the standard formulation in equations \eqref{1}, the transport operator described in the difference formulation \eqref{3}, as known as the classical reference field formulation, remains unchanged. This means particles still travel at the speed of light, collisions and absorption 
still occur, and the opacity stays the same in the difference field. In the classical reference field, the thermal emission term disappears while additional source terms appear that involve temporal and spatial derivatives of the blackbody radiation intensity $B$ based on the current temperature. When these derivative terms are not much larger than the thermal emission term (when the problem is close to steady-state and the temperature gradients are not too large), the Monte Carlo method for thermal radiative transfer based on the classical reference field formulation can significantly reduce statistical errors. Suppose the problem is in a state of radiation equilibrium, i.e., when the radiation temperature equals the material temperature. In that case, we find that the statistical error of the Monte Carlo method for thermal radiative transfer based on the classical reference field is zero. At this point, the source emits zero particle energy at each time step. In contrast, the Monte Carlo method based on the classical formulation still exhibits statistical errors.

%% file: grf_method.tex
\subsection{The generalized reference field method}
Luu et al. \cite{luu2010generalized} found that although the thermal emission term vanishes, the spatial and temporal derivative terms related to the blackbody radiation intensity still resulted in relatively high statistical noise when a crucial difference existed between the radiation field intensity and the blackbody radiation field intensity. To address this, they proposed the generalized reference field method. The central idea of the generalized reference field method is no longer to restrict the reference field to being the blackbody radiation field at the current time. Instead, it can be taken as the blackbody radiation field at any arbitrary temperature.

Let $\widetilde{B}$ as a reference intensity field, with the corresponding reference field temperature defined as $\widetilde{T}$, the difference field intensity is given by $D = I - \widetilde{B}$, and the thermal radiation transport equations \eqref{3} are reformulated as:
\begin{subequations}
\begin{align}
& \frac{1}{c}\frac{\partial D}{\partial t}+\mu\frac{\partial D}{\partial x}+\sigma D=\sigma(B-\widetilde{B}) -\frac{1}{c}\frac{\partial  \widetilde{B}}{\partial t}-\mu\frac{\partial \widetilde{B}}{\partial x}+\frac{1}{4\pi}Q,
\\
& \frac{\partial U_m}{\partial t}=2\pi\int_0^{\infty}\int_{-1}^1 \sigma D\dd\mu\dd\nu-2\pi\int_0^{\infty}\int_{-1}^1 \sigma(B-\widetilde{B}) \dd\mu\dd\nu.
\end{align}
\label{5}
\end{subequations}
The initial and boundary conditions can be analogously specified as those for equations \eqref{3} and \eqref{4}. In this paper, we refer to equations \eqref{5} as the generalized reference field formulation of the thermal radiation transport equations.

Based on straightforward observation, Luu primarily investigated two specific types of reference fields:
\begin{align*}
\widetilde{B}_1(\nu, T(x,t))& =B(\nu, T(x,t)),\\
\widetilde{B}_2(\nu, T(x,t))& =B(\nu, T(x,t^n)),
\end{align*}
here $t\in(t^n,t^{n+1}]$. The first type of reference field is the same as the classic reference field, while the second type of reference field is taken as the blackbody radiation field at the beginning of each time step.

Although Luu introduced the concept of the generalized reference field method, he did not provide a selection strategy for the reference field. In regions with significant temperature gradients, neither of these reference fields can significantly reduce statistical noise, and there may even be instances where the statistical noise is greater than in the standard formulation \cite{2021The,cleveland2010extension}. Additionally, the reference field method introduces negative energy weights in computations, which causes numerical instability. For these reasons, the reference field method has received little attention for a long period.

\section{The global optimal reference field method}
Although the definition of the generalized reference field method is introduced above, a vital issue remains unresolved —— how to select the reference field? In our previous work \cite{2021The}, we answered this question by proposing the global optimal reference field method and transforming the original problem of determining the global optimal reference field into an optimization problem. In the following subsections, we will briefly review the IMC method for thermal radiation transport problems, propose the global optimal reference field concept, and introduce the connection between the global reference field and the target convex optimal problem. 

\subsection{The Implicit Monte Carlo method}
In this section, we derive the IMC method based on the generalized reference field formulation of the thermal radiation transport equation in \eqref{5}. 
In the Symbolic Implicit Monte Carlo(SIMC) method, the reference field can be unknown. The fully implicit scheme is applied to the material temperature, which is determined by the solution of a nonlinear equation system at the next time step \cite{BROOKS1989433}.
Although the IMC method is essentially a semi-implicit scheme, it is solved explicitly. If the reference field is unknown, the IMC method for thermal radiative transfer under the generalized reference field formulation is challenging to solve explicitly. In this work, We ignore this situation and only consider the case where the reference field is known at each time step.
We assume that at the time $t^{n}$, the reference field intensity is already given and is denoted as  $\widetilde{B}^n(\nu, \widetilde{T}^{n})$, which represents the blackbody radiation intensity for a specific reference temperature $\widetilde{T}^n$.

Generally, the radiation field is characterized by the radiation intensity, which includes the following variables: temporal variable, spacial variable, angular variable, and frequency variable. To simplify the problem, we assume the reference field is the blackbody radiation field and is relative to the material temperature on each grid cell. Based on this assumption, we can characterize the reference field using the material temperature or equilibrium radiation energy density on each grid cell, which greatly simplifies the problem. Furthermore, we assume that the reference field intensity is known at the beginning of each time step and does not change during each time step. For consistency with the context, we define the equilibrium radiation energy density corresponding to the reference field as $\widetilde{U}_r^n=a(\widetilde{T}^n)^4$.

Following the approach used in the IMC method for the standard formulation, we make the following approximations within a time step $[t^n,t^{n+1}]$:
\begin{subequations}
\begin{align}
& \sigma^n=\sigma(x,\nu,T(x, t^n))\approx \sigma(x,\nu,T(x, t)),\\
& b^n=b(\nu, T(x,t^n))\approx b(\nu, T(x,t)),\\
& \sigma_p^n=\int_0^{\infty}\sigma^nb^n\dd\nu\approx \int_0^{\infty}\sigma(x,\nu, T(x,t))b(\nu,T(x,t))\dd\nu,\\
& \beta^n=\beta(T(x,t^n))=\frac{\partial U_r}{\partial U_m}|_{x,t^n}\approx \beta(T(x, t)).
\end{align}
\label{6}
\end{subequations}
Here, $t\in (t^n,t^{n+1}]$. Plugging the equations \eqref{6} into equations \eqref{5}, we rewrite the equations \eqref{5} in the following form:
\begin{subequations}
\begin{align}
& \frac{1}{c}\frac{\partial D}{\partial t}+\mu\frac{\partial D}{\partial x}+\sigma^n D=\frac{c}{4\pi}\sigma^n(U_rb-\widetilde{U}_r^n\widetilde{b}^n)- \frac{1}{c}\frac{\partial  \widetilde{B}^n}{\partial t}-\mu\frac{\partial \widetilde{B}^n}{\partial x}+\frac{1}{4\pi}Q^n,\label{7a}\\
& \frac{\partial U_r}{\partial t}=\beta^n\left(
\int_0^{\infty}\int_{-1}^1 \sigma^n 2\pi D\dd\mu\dd\nu-\int_0^{\infty}\int_{-1}^1 \sigma^n2\pi(B-\widetilde{B}^n) \dd\mu\dd\nu
\right),\label{7b}
\end{align}
\label{7}
\end{subequations}
in which, we utilize the relationship in \eqref{def:B}:
\begin{align*}
B(\nu, T(x,t))=\frac{c}{4\pi}U_r(T(x,t))b(\nu, T(x,t)).
\end{align*}
From the definition in \eqref{def:B} and \eqref{def:sigp}, it follows:
\begin{equation}
\begin{aligned}
2\pi\int_0^{\infty}\int_{-1}^1 \sigma^n B \dd\mu\dd\nu=\sigma_p^ncU_r,
\end{aligned}
\label{8}
\end{equation}
Similar to \eqref{def:sigp}, we define the mean opacity at reference temperature $\widetilde{T}$ in the reference field, denoted as $\widetilde{\sigma}_p$:
\begin{equation}
\begin{aligned}
\widetilde{\sigma}_p^n=\int_0^{\infty}\sigma^nb(\nu, \widetilde{T}(x,t^n))\dd\nu.
\end{aligned}
\label{9}
\end{equation}
Substituting equations \eqref{8} and \eqref{9} into equation \eqref{7b}, we can obtain:
\begin{align}
\frac{\partial U_r}{\partial t}=\beta^n\left(
\int_0^{\infty}\int_{-1}^1 \sigma^n2\pi D\dd\mu\dd\nu-\sigma_pcU_r+\widetilde{\sigma}_p^nc\widetilde{U}_r^n
\right).
\label{10}
\end{align}
The integral average of a parameter over a time step is defined as:
\begin{align*}
\overline{(\cdot)}=\frac{1}{\Delta t^n}\int_{t^n}^{t^{n+1}}(\cdot)\dd t,
\end{align*}
Since the reference field is given at the beginning of each time step and the reference field does not change within a time step, the equation \eqref{10} over a time step yields:
\begin{align}
\frac{U_r^{n+1}-U_r^n}{\Delta t^n}=\beta^n\left(
\int_0^{\infty}\int_{-1}^1 \sigma^n2\pi \overline{D}\dd\mu\dd\nu-\sigma_pc\overline{U}_r+\widetilde{\sigma}_p^nc\widetilde{U}_r^n 
\right).\label{11}
\end{align}
With the same assumption in the standard IMC method, we approximate $\overline{U_r}$ as:
\begin{align}
\overline{U_r}	=\alpha  U_r^{n+1}+(1-\alpha ) U_r^{n}, \quad 0.5 \leq \alpha \leq 1.\label{12}
\end{align}
Combining equations \eqref{11} and \eqref{12}, we can obtain:
\begin{align}
\overline{U_r}=f^n U_r^n+\frac{1}{c\sigma_p^n}\int_0^{\infty}\int_{-1}^1(1-f^n)\sigma^n2\pi\overline{D}\dd \mu\dd \nu +(1-f^n)\frac{\widetilde{\sigma}_p^n}{\sigma_p^n}\widetilde{U}_r^n,\label{13}
\end{align}
where $f^n$ is the Fleck factor for the current time step, which is given as:
\begin{align*}
f^n=\frac{1}{1+\alpha c\beta^n\sigma_p^n\Delta t^n}.
\end{align*}
Finally, approximating the time-integrated averages in equation \eqref{13} with instantaneous values, i.e., $\overline{U}_r(x)\approx U_r(T(x,t))$ and $\overline{D}(x,\mu,\nu) \approx D(x,t,\mu,\nu)$, we obtain the analytical expression of $U_{r}(x,t)$:
\begin{align}
U_r(x,t)=f^n U_r^n+\frac{1}{c\sigma_p^n}\int_0^{\infty}\int_{-1}^1(1-f^n)\sigma^n2\pi D(x,t,\mu,\nu)\dd \mu\dd\nu +(1-f^n)\frac{\widetilde{\sigma}_p^n}{\sigma_p^n}\widetilde{U}_r^n, \quad t\in (t^{n}, t^{n+1}]. \label{14}
\end{align}
For the convenience of further deduction, we define the effective absorption opacity $\sigma_{ea}^{n}$ and the effective scattering opacity $\sigma_{es}^{n}$ at time $t = t^{n}$:
\begin{subequations}
\begin{align}
& \sigma_{ea}^n=f^n\sigma^n,\\
& \sigma_{es}^n=(1-f^n)\sigma^n.
\end{align}
\label{15}
\end{subequations}
Combining equations \eqref{14} and \eqref{15} and substituting them into equation \eqref{7a},  the linearized radiation transport equation is obtained as follows:

\begin{equation}\label{16}
   \begin{split}
       \frac{1}{c}\frac{\partial D}{\partial t}+\mu\frac{\partial D}{\partial x}+\sigma^n D=&\frac{c}{4\pi}\sigma_{ea}^nb^nU_r^n
    +\frac{c}{4\pi}\sigma_{es}^nb^n\frac{\widetilde{\sigma}_p^n}{\sigma_p^n}\widetilde{U}_r^n-\frac{c}{4\pi}\sigma^n\widetilde{b}^n\widetilde{U}_r^n\\
    & +\frac{\sigma^nb^n}{2\sigma_p^n}\int_{0}^{\infty}\int_{-1}^1\sigma_{es}^nD\dd \mu \dd\nu
    -\frac{1}{c}\frac{\partial  \widetilde{B}}{\partial t}-\mu\frac{\partial \widetilde{B}}{\partial x}+\frac{1}{4\pi}Q^n.
   \end{split} 
\end{equation}
The standard Monte Carlo method is applied for computation since the source terms of the right side in \eqref{16} are already known. Additionally, to obtain the internal energy density of the material at time $t^{n+1}$, we need to track the energy of particles absorbed by the material in each time step. Specifically, equation \eqref{16} can be viewed as a particle transport process involving absorption and scattering, with an absorption coefficient $\sigma_{ea}$. Therefore, the energy of particles absorbed by the material in each time step can be expressed as:
\begin{align}
\int_{t^n}^{t^{n+1}}\int_0^{\infty}\int_{-1}^1\sigma_{ea}2\pi D\dd\mu\dd\nu\dd t. \label{17}
\end{align}
In the Monte Carlo method, we obtain the statistical average of expression \eqref{17} by tracking the energy deposition at each time step. However, this process inevitably leads to statistical errors, which is an inherent limitation of the Monte Carlo method \cite{FLECK1971313,WOLLABER200808}.

With the help of the denotation defined above, the energy balance equation \eqref{7b} can be rewritten in the following equivalent form:
\begin{align}
\frac{\partial U_m}{\partial t}=\int_0^{\infty}\int_{-1}^1\sigma^n 2\pi D\dd\mu\dd\nu-\sigma_p^ncU_r+\int_0^{\infty}\int_{-1}^1\sigma^n2\pi \widetilde{B}\dd\mu\dd\nu.\label{18}
\end{align}
The equation \eqref{14} is introduced to approximate the instantaneous value $U_{r}$, which is then substituted into the energy balance equation \eqref{18} to obtain
\begin{align}
\frac{\partial U_m}{\partial t}=\int_0^{\infty}\int_{-1}^1\sigma_{ea}^n2\pi D\dd\mu\dd\nu-f^n\sigma_p^ncU_r^n+\int_0^{\infty}\int_{-1}^1\sigma_{ea}^n2\pi \widetilde{B}\dd\mu\dd\nu.\label{19}
\end{align}
Integrating equation \eqref{19} from $t^{n}$ to $t^{n+1}$, it follows:
\begin{equation}\label{20}
U_m^{n+1}-U_m^n =\int_{t^n}^{t^{n+1}}\int_0^{\infty}\int_{-1}^1\sigma_{ea}^n2\pi D\dd\mu\dd\nu\dd t
 -f^nc\sigma^n U_r^n\Delta t^{n}+\int_{t^n}^{t^{n+1}}\int_0^{\infty}\int_{-1}^1\sigma_{ea}^n2\pi \widetilde{B}\dd\mu\dd\nu\dd t.
\end{equation}
In above equation, the first term on the right-hand side is determined statistically, while the second and third terms can be directly determined from their values at $t^n$. Therefore, we can straightforward compute $U_m^{n+1}$ and consequently calculate $T^{n+1}$.

Compared to the standard formulation, the source term in the IMC equations \eqref{16} under the generalized reference field formulation is written as:
\begin{align}
\frac{c}{4\pi}\sigma_{ea}^nb^nU_r^n+\frac{c}{4\pi}\sigma_{es}^nb^n\frac{\widetilde{\sigma}_p^n}{\sigma_p^n}\widetilde{U}_r^n -\frac{c}{4\pi}\sigma^n\widetilde{b}^n\widetilde{U}_r^n+\frac{\sigma^nb^n}{2\sigma_p^n}\int_0^{\infty}\int_{-1}^1\sigma_{es}^nD\dd \mu \dd\nu
-\frac{1}{c}\frac{\partial  \widetilde{B}}{\partial t}-\mu\frac{\partial \widetilde{B}}{\partial x}+\frac{1}{4\pi}Q^n.\label{21}
\end{align}
Here, the first three terms represent the thermal emission terms based on the reference field. The temporal, spatial, angular, and frequency probability density functions for the emitted particles are the same due to gray approximation. As for the frequency-dependent problem,  the frequency probability density functions are not identical among particles.
Taking the $i$-th cell as an example, the first two terms can be written in the following form:
\begin{align}
\left[
\left(f^n\sigma_{p,i}^nU_{r,i}^n+(1-f^n)\sigma_{p,i}^n\widetilde{U}_{r,i}^n\right)c\Delta t^nV_i
\right]
\left(
\frac{1}{V_i}\frac{1}{\Delta t^n}\frac{1}{2}\frac{\sigma_i^nb_i^n}{\sigma_{p,i}^n}
\right).
\label{22}
\end{align}
Here, we use \( [ \cdot ] \) to denote the total energy emitted by the source in each time step on the $i$-th grid cell, and \( (\cdot) \) to denote the probability density function for the spatial, temporal, angular, and frequency of particles emitted from the source. Similarly, the third term can be written as:
\begin{align*}
\left[
-\widetilde{\sigma}_{p,i}^n\widetilde{U}_{r,i}^nc\Delta t^nV_i
\right]
\left(
\frac{1}{V_i}\frac{1}{\Delta t^n}\frac{1}{2}\frac{\sigma_i^n\widetilde{b}_i^n}{\widetilde{\sigma}_{p,i}^n}
\right).
\end{align*}
In practical applications, the material temperature associated with the reference field is often not significantly different from that at time \( t^n \). Consequently, we can approximate \( \widetilde{b}^n \) with \( b^n \) , which results in \( \widetilde{\sigma}_p^n \) being equivalent to \( \sigma_p^n \). Under these conditions, the first three terms can be simplified as follows:
\begin{align}
\left[
f^n\sigma_{p,i}^n(U_{r,i}^n-\widetilde{U}_{r,i})c\Delta t^nV_i
\right]
\left(
\frac{1}{V_i}\frac{1}{\Delta t^n}\frac{1}{2}\frac{\sigma_i^nb_i^n}{\sigma_{p,i}^n}
\right).
\label{23}
\end{align}
The fourth term corresponds to the scattering source, which does not emit particles. When the collision of particles is a scattering event, the energy weight and position of the particles do not change, but their angles and frequencies undergo alterations. From this perspective, it can be inferred that the angles of the particles are uniformly distributed.

The fifth term represents the temporal partial derivative term of the reference field, which can be written as:
\begin{align}
\left[
(\widetilde{U}_{r,i}^{n-1}- \widetilde{U}_{r,i}^{n})V_i
\right]
\left(
\frac{1}{V_i}\delta(t-t^n)\frac{1}{2}f(\nu,\widetilde{T}_i^{n-1}, \widetilde{T}_i^{n})
\right),\label{24}
\end{align}
here the function $f$ can be written as:
\begin{equation}
f(\nu, T_1, T_2)=\frac{T_1^4b(\nu, T_1)-T_2^4b(\nu, T_2)}{T_1^4-T_2^4}.
\label{25}
\end{equation}

The sixth term represents the spatial partial derivative term of the reference field. Similarly, adopting the idea of electron pairs and considering the positive \( \mu \) direction as an example, the spatial derivative term can be written in the following form:

\begin{align}
\left[
\frac{1}{4}c\Delta t^n\left( (\widetilde{T}_{l(j)}^{n})^4-(\widetilde{T}_{r(j)}^{n})^4\right)
\right]
\left(
\delta(x-x_j)\frac{1}{\Delta t^n}2\mu f(\nu, \widetilde{T}_{l(j)}^{n}, \widetilde{T}_{r(j)}^{n})
\right),\label{26}
\end{align}

Here, \( x_j \) represents a specific interface of the \( i \)-th cell, \( l(j) \) denotes the grid index to the left of the interface, and \( r(j) \) denotes the grid index to the right of the interface. For more details on the implicit Monte Carlo method in solving thermal radiative transport equation, we recommend readers refer to \cite{2021The}.

\subsection{The global optimal reference field}
In this part, we will further analyze the reason for statistical noise in the Monte Carlo method. Firstly, we present the standard and reference field formulation of the IMC method. Secondly, we reveal that the difference of the statistical noise in the two formulations comes from the Monte Carlo particle weight. Let the computational domain be \([0, X]\). In a single time step \((t^n, t^{n+1}]\), the one-dimensional IMC equations under the standard formulation can be written as:
\begin{subequations}
\begin{align}
& \frac{1}{c}\frac{\partial I}{\partial t} + \mu\frac{\partial I}{\partial x} + \sigma^n I = \frac{c}{4\pi}\sigma_{ea}^n b^n U_r^n + \frac{1}{2}\frac{\sigma^nb^n}{\sigma_p^n}\int_0^{\infty}\int_{-1}^1\sigma_{es}^n I \, d\mu \, d\nu + \frac{1}{4\pi}Q,\label{27a}\\
& \frac{\partial U_m}{\partial t} = \int_0^{\infty}\int_{-1}^1 \sigma_{ea}^n 2\pi I \, d\mu \, d\nu - f^n \sigma_p^n c U_r^n,\label{27b}
\end{align}
\label{27}
\end{subequations}
where the initial conditions are:
\begin{subequations}
\begin{gather*}
I(x, t^n, \mu, \nu) = I^{i}(x, \mu, \nu),\\
T(x, t^n) = T^{i}(x),
\end{gather*}
\end{subequations}
and the boundary conditions are:
\begin{subequations}
\begin{gather*}
I(0, t, \mu, \nu) = I^l(t, \mu, \nu), \quad 0 \le \mu \le 1,\\
I(X, t, \mu, \nu) = I^r(t, \mu, \nu), \quad -1 \le \mu \le 0.
\end{gather*}
\end{subequations}
Let \(\widetilde{B}^{n-1}\) denote the reference field in the time step \((t^{n-1}, t^{n}]\), and \(\widetilde{B}^n\) denote the reference field in the time step \((t^n, t^{n+1}]\). According to the derivation in the previous section, the one-dimensional IMC equations under the generalized reference field formulation can be written as:
\begin{subequations}
\begin{align}
\frac{1}{c}\frac{\partial D}{\partial t}+\mu\frac{\partial D}{\partial x}+\sigma^n D=\frac{c}{4\pi}\sigma_{ea}^nb^nU_r^n
+\frac{c}{4\pi}\sigma_{es}^nb^n\frac{\widetilde{\sigma}_p^n}{\sigma_p^n}\widetilde{U}_r^n-\frac{c}{4\pi}\sigma^n\widetilde{b}^n\widetilde{U}_r^n\notag\\
 +\frac{\sigma^nb^n}{2\sigma_p^n}\int_0^{\infty}\int_{-1}^1\sigma_{es}^nD\dd \mu\dd\nu
-\frac{1}{c}\frac{\partial  \widetilde{B}}{\partial t}-\mu\frac{\partial \widetilde{B}}{\partial x}+\frac{1}{4\pi}Q^n,\label{28a}\\
\frac{\partial U_m}{\partial t}=\int_0^{\infty}\int_{-1}^1\sigma_{ea}^n2\pi D\dd\mu\dd\nu-f^n\sigma_p^ncU_r^n+2\pi\int_0^{\infty}\int_{-1}^1\sigma_{ea}^n2\pi \widetilde{B}\dd\mu\dd\nu,
\label{28b}
\end{align}
\label{28}
\end{subequations}
where the initial conditions are:
\begin{gather*}
D(x, t^n, \mu, \nu) = I^{i}(x, \mu, \nu) - \widetilde{B}^{n-1}(x, \nu),\\
T(x, t^n, \mu, \nu) = T^{i}(x),
\end{gather*}
and the boundary conditions are:
\begin{gather*}
D(0, t, \mu, \nu) = I^l(t, \mu, \nu) - \widetilde{B}^n(x=0, \nu), \quad 0 \le \mu \le 1, \\
D(X, t, \mu, \nu) = I^r(t, \mu, \nu) - \widetilde{B}^n(x=X, \nu), \quad -1 \le \mu \le 0. 
\end{gather*}
An important characteristic of the IMC equations under the generalized reference field formulation is that the reference field may be discontinuous in time, which results from the requirement of a new reference field at the beginning of each time step. Additionally, it is evident that equation \eqref{27} and equation \eqref{28} are equivalent under these two formulations.

We now focus on the source of the statistical noise in the IMC method for the thermal radiation transport problem. In this method, we need to statistically track the energy deposition on spatial grids at each time step, which is the portion of the energy absorbed by the background material after particles collide. For the IMC method, the expression for this energy is:

\begin{equation*}
\int_{t^n}^{t^{n+1}}\int_0^{\infty}\int_{-1}^1\sigma_{ea}2\pi I \, d\mu \, d\nu \, dt.
\end{equation*}
For the IMC method under the generalized differential field formulation, the expression for this energy is:
\begin{equation*}
\int_{t^n}^{t^{n+1}}\int_0^{\infty}\int_{-1}^1\sigma_{ea}2\pi D \, d\mu \, d\nu \, dt.
\end{equation*}
In the following analysis, we ignore the differences in the spatial, temporal, angular, and frequency probability density functions of particles under the two formulations and only focus on the differences in the energy weights of the particles. We define the energy deposition density of the \(k\)-th Monte Carlo particle in a volume \(V\) grid cell during the time interval \((t^n, t^{n+1}]\) as \(g_{k}w_k/V\), where \(w_k\) is the energy weight of the \(k\)-th particle, and \(g_{k}\) is a random number that is either 0 or 1, representing the probability that the \(k\)-th Monte Carlo particle is absorbed in the grid cell. When \(g_k = 0\), the particle is not absorbed; when \(g_k = 1\), the particle is absorbed by the grid.

Assume the  $g_{k}$ follows a probability density distribution function $G_{k}$. The transport operators are consistent under both formulations, meaning that particles travel at the speed of light and the opacity of the background material is the same. Since we have ignored the differences in the spatial, temporal, angular, and frequency probability density functions of the particles at birth, this implies that the probability density distribution function $G_{k}$ is consistent under both formulations, even though we cannot directly obtain the expression for \(G_{k}\).

To distinguish between the two formulations, we use the subscript \(sf\) to denote quantities under the standard formulation and the subscript \(gf\) to denote quantities under the generalized reference field formulation. We assume that both formulations of the IMC method use $N_{p}$ particles. For the IMC method in the standard formulation, the change in material energy density in a time step \((t^n, t^{n+1}]\) is:
\begin{equation}\label{29}
U_{m,sf}^{n+1} - U_{m,sf}^{n} = \frac{1}{V}\sum_{k=1}^{N_p} w_{k,sf} g_{k,sf} - cf^n \sigma_{p,sf}^n U_{r,sf}^n \Delta t^n.
\end{equation}
Similarly, for the IMC method in the generalized reference field formulation, the change in material energy density in single time step \((t^n, t^{n+1}]\) is:
\begin{equation}\label{30}
U_{m,gf}^{n+1} - U_{m,gf}^{n} = \frac{1}{V}\sum_{k=1}^{N_p} w_{k,gf} g_{k,gf} - cf^n \sigma_{p,gf}^n U_{r,gf}^n \Delta t^n + f^n c \widetilde{\sigma}_{p,gf}^n \widetilde{U}_r^n \Delta t^n.
\end{equation}
Here, \(w_{k,sf}\) and \(w_{k,gf}\) represent the energy weights of the IMC particles under the standard and generalized reference field formulations, respectively. \(g_{k,sf}\) and \(g_{k,gf}\) are random numbers that follow the same probability distribution function \(G_k\). Additionally, from equation sets \eqref{27} and \eqref{28}, we know that:

\begin{align*}
U_{m,sf}^{n} & = U_{m,gf}^{n},\\
\sigma_{p,sf}^n & = \sigma_{p,gf}^n,\\
U_{r,sf}^{n} & = U_{r,gf}^{n}.
\end{align*}
Observing equations \eqref{29} and \eqref{30}, we find that the statistical error originates from the first term on the right-hand side of each equation, since the other terms are deterministic. Therefore, we have:

\begin{subequations}\label{work1:eq23}
\begin{align}
\mathrm{Var}(U_{m,sf}^{n+1}) &= \frac{1}{V}\mathrm{Var}\left(\sum_{k=1}^{N_p}w_{k,sf}g_{k,sf}\right),\\
\mathrm{Var}(U_{m,gf}^{n+1}) &= \frac{1}{V}\mathrm{Var}\left(\sum_{k=1}^{N_p}w_{k,gf}g_{k,gf}\right),
\end{align}
\end{subequations}
where ``\(\mathrm{Var}\)" denotes the variance. Since the Monte Carlo particles are mutually independent, $\left\{g_{k, gf}\right\}_{k=1}^{N_{p}}$ and $\left\{ g_{k, sf}\right\}_{k=1}^{N_{p}}$ are also mutually independent as well. Additionally, the absolute values of the energy weights of all Monte Carlo particles are generally taken to be consistent, which follows:
\begin{subequations}\label{work1:eq24}
\begin{align}
w_{k,sf} &= \frac{1}{N_p}E_{tot,sf}^n, \quad k=1,2,\ldots,N_p,\\
w_{k, gf} &= \frac{1}{N_p}E_{tot,gf}^n, \quad k=1,2,\ldots,N_p.
\end{align}
\end{subequations}
Here, \(E_{tot,sf}^n\) and \(E_{tot,gf}^n\) represent the total energy of all Monte Carlo particles emitted from the sources in a time step \((t^n, t^{n+1}]\) under the standard and generalized reference field formulations, respectively. For the standard formulation, the energy weights of the Monte Carlo particles are all positive. 
However, some ``energies" emitted from specific sources might be negative for the generalized reference field formulation. At the start of the emission of Monte Carlo particles, these negative energies cannot cancel out the positive energies. Therefore, once a source emits ``negative energy" $(E'\le 0)$, the number of particles from this source is $\frac{|E'|}{w_{gf}}$ and the energy weight of each Monte Carlo particle is $-w_{gf}$. 

With the assumption of constant energy weights of particles, we derive the following equations from the \eqref{work1:eq23} and \eqref{work1:eq24}:
\begin{subequations}\label{work1:eq25}
\begin{align}
\mathrm{Var}(U_{m,sf}^{n+1}) &= \frac{1}{V}\frac{1}{N_p^2}(E_{tot,sf}^n)^2\sum_{k=1}^{N_p}\mathrm{Var}(g_{k,sf}),\\
\mathrm{Var}(U_{m,gf}^{n+1}) &= \frac{1}{V}\frac{1}{N_p^2}(E_{tot,gf}^n)^2\sum_{k=1}^{N_p}\mathrm{Var}(g_{k,gf}).
\end{align}
\end{subequations}
Since \(g_{k,sf}\) and \(g_{k,gf}\) follow the same probability density distribution function \(G_k\), we have:
\begin{equation}
\mathrm{Var}(g_{k,sf}) = \mathrm{Var}(g_{k,gf}).\label{work1:eq26}
\end{equation}
From equations \eqref{work1:eq25} and \eqref{work1:eq26}, we can infer that if \(\mathrm{Var}(U_{m,gf}^{n+1}) < \mathrm{Var}(U_{m,sf}^{n+1})\) is attained, it is required the following condition must be true:
\begin{align*}
E_{tot,gf}^n < E_{tot,sf}^n.
\end{align*}
Furthermore, the lower the value of \(E_{tot,gf}^n\), the smaller the statistical error will be in the IMC method that employs the generalized reference field formulation. It should be noted that this conclusion was reached without considering the discrepancies in the spatial, temporal, angular, and frequency probability density functions of the particles at the time of emission between the two formulations.

As mentioned above, statistical error is associated with the total energy of all Monte Carlo particles. To obtain a lower statistical error, it is imperative to ensure a lower total energy. We then present the total energy of standard and generalized reference field formulations in detail. For the classical IMC method, we have:

\begin{equation}
E_{tot,sf}^n = E_{B,sf}^n + E_{Q,sf}^n + E_{IC,sf}^n + E_{R,sf}^n,
\label{work1:eq27}
\end{equation}
where \(E_{B,sf}^n\), \(E_{Q,sf}^n\), \(E_{IC,sf}^n\), and \(E_{R,sf}^n\) represent the energies emitted by the boundary source, external independent source, residual source from the previous time step, and thermal emission source in a time step, respectively. The analytical expressions of those terms are given by:

\begin{subequations}\label{work1:eq28}
\begin{align}
&E_{B,sf}^n = \int_0^{\infty}\int_0^1\int_{t^n}^{t^{n+1}}\mu 2\pi I^{l}(t,\nu,\mu) \, dt \, d\mu \, d\nu + \int_0^{\infty}\int_{-1}^0\int_{t^n}^{t^{n+1}}|\mu| 2\pi I^{r}(t,\nu,\mu) \, dt \, d\mu \, d\nu,\\
&E_{Q,sf}^n = \int_0^{\infty}\int_{-1}^1\int_{t^n}^{t^{n+1}}\int_0^X \frac{1}{2}Q(x,t,\nu) \, dx \, dt \, d\mu \, d\nu,\\
&E_{IC,sf}^n= \int_0^{\infty}\int_{-1}^1\int_0^X 2\pi I^{i}(x,\mu,\nu) \, dx \, d\mu \, d\nu,\\
&E_{R,sf}^n= \int_0^{\infty}\int_{t^n}^{t^{n+1}}\int_0^X \int_{-1}^1 \frac{1}{2}\sigma_{ea}^n b^n c U_r^n \, dx \, dt \, d\nu = c\Delta t^n \int_0^X f^n \sigma_p^n U_r^n \, dx.
\end{align}
\end{subequations}
For the IMC equations under the generalized reference field formulation, according to equation \eqref{work1:eq27}, we can write the total energy emitted by the sources within a time step in the following expression:

\begin{equation}
E_{tot, gf}^n = E_{B,gf}^n + E_{Q,gf}^n + E_{R,gf}^n + E_{IC,gf}^n + E_{dt,gf}^n + E_{dxl, gf}^n + E_{dxr, gf}^n.
\label{work1:eq29}
\end{equation}
We will analyze each term on the right-hand side of the equation \eqref{work1:eq29} item by item.

$\bullet \textbf{ The boundary source term}$ $E_{B, gf}^{n}$

The first quantity \(E_{B,gf}^n\), represents the energy emitted by a boundary source within a single time step. In this work, we introduce several specific types of boundary conditions. First, a Planck source at constant temperature on the boundary, which is common in numerical experiments. In this case, the reference field temperature of the grid cell just outside the boundary can be set to the Planck source's temperature. The boundary source can then be considered as a spatial derivative source at the boundary. 

Using the concept of electron pairs but focusing only on particles entering the computational domain, we take the left boundary as an example and describe the energy emitted from this source within a time step. Suppose the temperature of the Planck source is \(\widetilde{T}_0\), which we define as the reference field temperature for the grid cell just outside the left boundary, while the reference field temperature of the first grid cell inside the boundary is \(\widetilde{T}_1\). Thus, the energy flowing into the left boundary within a single time step is given by:
\begin{equation}
E_{Bl, gf}=\lim_{\epsilon_x\rightarrow 0}\int_{-\epsilon_x}^{\epsilon_x}\int_{t^n}^{t^{n+1}}\int_0^1\int_0^{\infty}  -2\pi \mu\frac{\partial \widetilde{B}}{\partial x} \dd\nu\dd\mu\dd t\dd x
=\frac{1}{4}ac\Delta t^n\left(
\widetilde{T}_0^4-\widetilde{T}_1^4
\right).
\label{work1:eq30}
\end{equation}
Here, \(\epsilon_x > 0\) is any arbitrarily small positive number. The result of expression\eqref{work1:eq30} might be negative, meaning that the energy weight of the Monte Carlo particles entering the boundary is negative. However, when calculating the total energy emitted by all sources within one time step, this part of the energy is taken as its absolute value and included in the total energy. Similarly, the total energy flowing into the right boundary can be defined as \(E_{Br, gf}\) within a time step, which yields the total boundary energy:
\begin{align*}
E_{B,gf}^n=\vert E_{Bl, gf}\vert +\vert E_{Br, gf}\vert.
\end{align*}
Secondly, suppose the left boundary has a reflective boundary condition, meaning there is no net radiative flux at the boundary. In that case, the reference field temperature of the grid cell just outside the left boundary can be set equal to the adjacent grid cell's reference field temperature inside the boundary. In this case, we obtain that  
\begin{align*}
E_{Bl,gf}^n=0.
\end{align*}
Thirdly, if the boundary source is not a Planck source, the reference field temperature of the grid cell just outside the left boundary is typically set to zero, where \(D^l = I^l\). Therefore, the energy flowing into the left boundary within a time step is given by:
\begin{align*}
E_{Bl, gf}^n=\int_{t^n}^{t^{n+1}}\int_0^1\int_0^{\infty}2\pi I^l(t,\mu,\nu)\dd \nu\dd \mu\dd t.
\end{align*}
Fourthly, when the left boundary is a vacuum boundary condition, it can be treated as a Planck source with zero temperature. In this case, the calculation process is similar to the first case.

$\bullet \textbf{ The independent external source term}$ $E_{Q, gf}^{n}$

The second quantity $E_{Q,gf}^n$, represents the energy emitted by an independent external source within a single time step. From equation \eqref{7a}, we have:
\begin{equation}
E_{Q,gf}^n=E_{Q,sf}^n=\int_0^{\infty}\int_{-1}^1\int_{t^n}^{t^{n+1}}\int_0^X\frac{1}{2}Q(x,t,\nu)\dd x\dd t\dd\mu\dd\nu.
\end{equation}

$\bullet \textbf{ The thermal emission term}$ $E_{R, gf}^{n}$

The third quantity \(E_{R,gf}^n\) represents the energy emitted by thermal emission within a single time step, which is expressed as:
\begin{equation}
E_{R,gf}^n=
2\pi\int_{t^n}^{t^{n+1}}\int_{-1}^1\int_0^{\infty}\left\vert
\frac{c}{4\pi}\sigma_{ea}^nb^nU_r^n+\frac{c}{4\pi}\sigma_{es}^nb^n\frac{\widetilde{\sigma}_p^n}{\sigma_p^n}\widetilde{U}_r^n -\frac{c}{4\pi}\sigma^n\widetilde{b}^n\widetilde{U}_r^n
\right\vert
\dd\nu\dd\mu\dd x.
\label{work1:eq31}
\end{equation}
With the same assumption of minor difference between material temperature within a single time step, we also apply the approximation $\widetilde{b}^n=b^n$, which simplifies the formulation of $E_{R,gf}$ as:
\begin{equation}\label{work1:eq31b}
E_{R,gf}^n =
\int_{t^n}^{t^{n+1}}\int_{-1}^1\int_0^{\infty}\frac{c}{2}\sigma_{ea}^nb^n\left\vert U_r^n-\widetilde{U}_r^n\right\vert\dd\nu\dd\mu\dd x.
\end{equation}

$\bullet \textbf{ The residual source term}$ $E_{IC, gf}^{n}$

The fourth term $E_{IC,gf}^n$ represents the energy of the residual source from the previous step. The state of all particles in the residual source is known, thus its energy can be calculated as:
\begin{equation}
E_{IC,gf}^n=\sum_{l=1}^{N_{IC,gf}}\vert w_{l}'\vert,
\label{work1:eq32}
\end{equation}
where $N_{IC,gf}$ is the total number of particles in the residual source, and $w_{l}'$ is the energy weight of the $l$-th particle in the residual source.

$\bullet \textbf{ The temporal derivative source term}$ $E_{dt,gf}^n$

The fifth term $E_{dt,gf}^n$ denotes the energy emitted by the time derivative source within a time step. Since the reference field is given at the beginning of each time step and remains constant throughout, the energy emitted by this source can be expressed as:
\begin{equation}\label{work1:eq33}
\begin{aligned}
E_{dt,gf}^n&=\lim_{\epsilon_t\rightarrow 0}\int_0^X\int_{t^n-\epsilon_t}^{t^n+\epsilon_t}\int_{-1}^1\int_0^{\infty}\left\vert-\frac{2\pi}{c}\frac{\partial \widetilde{B}}{\partial t}\right\vert\dd\nu\dd\mu\dd t\dd x\\
&=\int_0^X\left\vert 
\widetilde{U}_r^n-\widetilde{U}_r^{n-1}
\right\vert\dd x.
\end{aligned}
\end{equation}

$\bullet \textbf{ The spatial derivative source term}$ $E_{dxr, gf}^{n}, E_{dxl,gf}^{n}$

The last two terms, $E_{dxr, gf }^n$ and $E_{dxl, gf }^n$, are the energy emitted by the spatial derivative source in the positive $\mu$ direction and the negative $\mu$ direction for each time step, respectively. In this work, the reference field is approximated by piece-wise constants in space, meaning it remains constant within each grid cell. Thus, the particles emitted by this source are located at the grid interfaces. For a particular interface $x_j$, the energy emitted by the spatial derivative source in the positive $\mu$ direction for each time step is:

\begin{equation}
E_{dxr, gf,j+ }^n=\lim_{\epsilon_x\rightarrow 0}\int_{x_j-\epsilon_x}^{x_j+\epsilon_x}
\int_{t^n}^{t^{n+1}}\int_{0}^1\int_0^{\infty}-
2\pi\mu\frac{\partial \widetilde{B}^n}{\partial x}\dd \nu\dd\mu\dd x=\frac{1}{4}ac((\widetilde{T}_{l(j)}^n)^4-(\widetilde{T}_{r(j)}^n)^4)\Delta t^n.\label{work1:eq34}
\end{equation}

$E_{dxr, gf,j }^n$ may take negative values, however, we take the absolute value when calculating the total energy. Since the particles emitted by the spatial derivative source are all located at the grid interfaces, we have:
\begin{equation}
E_{dxr, gf }^n=\sum_{j}\vert E_{dxr, gf, j }^n\vert.\label{work1:eq35}
\end{equation}
According to the concept of electron pairs, it is easy to see that:
\begin{equation*}
E_{dxl, gf }^n=E_{dxr, gf }^n.
\end{equation*}

\subsection{The optimization problem}
In the preceding sections, we have established that the statistical noise is proportional to the total energy emitted from the source term as indicated in equation \eqref{work1:eq25}. A reduction in the total energy results in decreased statistical noise at each time step. Moreover, with the reference field defined at the commencement of each time step, we have dissected and elucidated the formulation of each component of these source terms. In the standard formulation, the energy emitted by the sources is predetermined once the initial and boundary conditions are specified. However, in the context of the reference field description, the emitted energy is contingent not only on the boundary and initial conditions but also on the selected reference field for the current time step.

In this subsection, we will reformulate the expression for the total energy as presented in equation \eqref{work1:eq29} and recast it into the framework of a standard optimization problem. Let $H$ denote the function associated with the generalized reference field $\widetilde{B}^{n}$, which is defined as follows:
\begin{equation}
H(\widetilde{B}^n)=E_{B,gf}^n+E_{Q,gf}^n+E_{R,gf}^n + E_{IC,gf}^n+E_{dt,gf}^n+E_{dxl, gf }^n+E_{dxr, gf}^n.\label{work1:eq36}
\end{equation}
It is obvious that \(H(\widetilde{B}^n) > 0\). Our goal is to find a reference field \(\widetilde{B}^n\) such that \(H(\widetilde{B}^n)\) takes on its minimum value.

It is a straightforward matter to demonstrate that when \( \widetilde{B} \) assumes finite values, \( H(\widetilde{B}^n) \) remains finite, and as \( \widetilde{B}^n \) approaches infinity, \( H(\widetilde{B}^n) \) also tends towards infinity. Consequently, there must be a reference field that minimizes \( H \), which we define as the "global optimal reference field" and denote by \( \ddot{B} \). It is clear that the quest for the global optimal reference field essentially boils down to an optimization problem, with the solution being the aforementioned global optimal reference field.


To simplify the problem, we assume the reference field is based on a blackbody radiation field related to a specific temperature, neglecting the frequency dependence. This reference field can be characterized by the temperature \(\widetilde{T}^n\) or the equilibrium radiation energy density \(\widetilde{U}_r^n = a(\widetilde{T}^n)^4\). Under this assumption, we employ a piecewise constant approximation in space, hypothesizing that each grid cell is homogeneous with a uniform cell size of $\Delta x$.

To further analyze Equation \eqref{work1:eq36}, we rewrite the left-hand side as follows:
\begin{equation}
H(\widetilde{B}^n)=H(\widetilde{U}_{r1}^n,\widetilde{U}_{r2}^n,\cdots,\widetilde{U}_{rN}^n)\label{work1:eq37}
\end{equation}
where the subscripts $1, 2, \ldots, N$ represent the grid indices. Then we calculate the right-hand side of equations \eqref{work1:eq36} term by term.


The first term, $E_{B,gf}^n$, represents the boundary source, which generally does not depend on the temperature in the spatial grids and thus can be considered a constant. The second term on the right-hand side, $E_{Q,gf}^n$, represents the external source term, which also does not depend on the temperature and is therefore regarded as a constant. The third term represents thermal emission, which is expressed in equation \eqref{work1:eq31b}. Under the assumption of a piecewise constant temperature approximation in space, this term is calculated as:

\begin{equation}
E_{R,gf}^n=c\Delta t\Delta x\sum_{i=1}^N f_i^n\sigma_i^n\left| U_{ri}^n-\widetilde{U}_{ri}^n\right|\label{work1:eq38}
\end{equation}
Here, $f_i^n$ represents the Fleck factor at the $i$-th grid at time $t^n$.

The fourth term on the right-hand side represents the residual source from the previous time step, which is independent of temperature and is treated as a constant. The fifth term on the right-hand side is the time derivative source, which under the assumption of piecewise constant approximation of temperature with respect to the spatial grid, can be written in the following form:
\begin{equation}
E_{dt,gf}^n = \Delta x \sum_{i=1}^N \left| \widetilde{U}_{ri}^n - \widetilde{U}_{ri}^{n-1} \right|
\label{work1:eq39}
\end{equation}
The sixth and seventh terms on the right-hand side are the spatial derivative sources, whose expressions is given by:
\begin{equation*}
E_{dxl, gf}^n = E_{dxr, gf}^n = \sum_{j} \left| \frac{1}{4} ac \left( (\widetilde{T}_{l(j)}^n)^4 - (\widetilde{T}_{r(j)}^n)^4 \right) \Delta t^n \right|.
\end{equation*}
Here, the subscript $l(j)$ denotes the value at the left end of the $j$-th interface, and the subscript $r(j)$ denotes the value at the right end of the $j$-th interface. Given the piecewise constant approximation of temperature, the above equation can be rewritten as:
\begin{equation}
E_{dxl, gf}^n = E_{dxr, gf}^n = \sum_{i=1}^{N+1} \left| \frac{1}{4} c (\widetilde{U}_{ri}^n - \widetilde{U}_{r(i-1)}^n) \Delta t^n \right|.
\label{work1:eq40}
\end{equation}
By combining equations \eqref{work1:eq38}, \eqref{work1:eq39}, and \eqref{work1:eq40}, we find that solving the global optimal reference field problem is transformed into solving the following equivalent unconstrained optimization problem:
\begin{equation}\label{work1:eq41}
\begin{aligned}
\displaystyle
\min \quad &f(\mathbf{z}) = \sum_{i = 1}^{N} \left(\alpha_{i}|z_{i} - \bar{\alpha}_{i}| + \beta_{i}|z_{i} - \bar{\beta}_{i}|\right) + \sum_{i=1}^{N+1}\gamma_{i} |z_{i} - z_{i-1}|,\\
\text{s.t.} \quad &\mathbf{z}= (z_1,z_2,\ldots,z_N)^T \in \mathbb{R}^N,
\end{aligned}
\end{equation}
Here, $z_{0}, z_{N+1}$ are the known constant determined by the boundary condition in the previous time step. The expressions for the other parameters are as follows:

$$
\begin{cases}
& \alpha_i = c\Delta t\Delta x f_i^n \sigma_i^n,\\
& \bar{\alpha}_i = U_{ri}^n,\\
& \beta_i = \Delta x,\\
& \bar{\beta}_i = \widetilde{U}_{ri}^{n-1},\\
& \gamma_i = \frac{1}{2}c\Delta t^n.
\end{cases}
$$

%% file: grf_scheme.tex
\section{Numerical method for solving the global optimal reference field}
In our previous work \cite{2021The}, we did not address the optimization problem directly. Instead, we proposed a physically sound and computationally economical solution that reduces statistical noise. However, while effective, this approach does not necessarily constitute the optimal solution to the aforementioned optimization problem. In this work, we will first briefly review the previous strategies for the target optimization problem \eqref{work1:eq41} and then will demonstrate that the optimization problem can be equivalently transformed into a linear programming problem with equality constraints and then solved efficiently by graph-neural-network-based solver MindOpt \cite{unknown2023}.

\subsection{The properties of the global optimal reference field}\label{property}
We first present some properties of the optimization problem in $\eqref{work1:eq41}$, which is essential to analyze the global optimal reference field.  
\begin{corollary}
\(f(\mathbf{z})\) is a non-negative convex function with respect to \(\mathbf{z}\), the optimal solution to the problem \eqref{work1:eq41} must exist.
\end{corollary}
It is easy to obtain this deduction. Furthermore, it is essential to point out the fact that the optimal solution $\mathbf{Z}^{\star} \in \left\{\bar{\alpha}_1,\cdots, \bar{\alpha}_N,  \bar{\beta}_1,\cdots, \bar{\beta}_n, z_0, z_{N+1}\right\}$, which is followed by the following theorem. 

\begin{theorem}\label{theorem}
    Consider the following unconstrained optimization problem:
    \begin{equation}\label{theorem:eq}
        \begin{aligned}
        \min \quad & f(z) = d_1' |z - d_1| + \cdots + d_n' |z - d_n|\\
        \text{s.t.} \quad & d_{i} \ge 0, \quad i \in \{1, 2, \cdots, n\},
        \end{aligned}
    \end{equation}
    where $d_{i}'$, $i \in \{1,2,\cdots,n\}$, are all positive coefficients. If the optimal solution $d^{\star}$ of \eqref{theorem:eq} exists and is unique, then it must satisfy the following condition:
    \begin{equation*}
        d^{\star} \in \{d_{1}, d_{2}, \cdots, d_{n}\}.
    \end{equation*}
\end{theorem}
We do not provide a proof for this theorem within this context. However, interested readers are directed to \cite{Boyd_Vandenberghe_2004} for more details. It is straightforward to demonstrate that a sufficient condition for the uniqueness of the optimal solution is:

\begin{equation}\label{unique}
    \pm d_{1} \pm d_{2} \pm \cdots \pm d_{n} \neq 0.
\end{equation}
In our practical problem, the sufficient condition above is generally satisfied. Using the Theorem.\ref{theorem}, we can deduce another important corollary below.

\begin{corollary}
 Let the set $\mathcal{Z}=\left\{\bar{\alpha}_1,\cdots, \bar{\alpha}_N,  \bar{\beta}_1,\cdots, \bar{\beta}_n, z_0, z_{N+1}\right\}$, if $\bm{Z}^{\star} \in \mathbb{R}^{N}$ is the optimal solution to problem \eqref{work1:eq41}, then $\bm{Z}_i^{\star} \in \mathcal{Z}$.
\end{corollary}
\begin{proof}
We will prove this by contradiction. Suppose $\mathbf{Z}^{\star} = (z^{\star}_1, z^{\star}_{2}, \ldots, z^{\star}_N)^T$ is an optimal solution to the optimization problem, and assume $z^{\star}_i \not\in \mathcal{Z}$. The objective function of the problem is given by:
\begin{equation}\label{prove:eq}
f(\mathbf{z}^{\star}) = \sum_{j=1}^{N} \left(\alpha_j |z^{\star}_j - \bar{\alpha}_j| + \beta_j |z^{\star}_j - \bar{\beta}_j|\right) + \sum_{j = 1}^{N+1} \gamma_j |z^{\star}_j - z^{\star}_{j-1}|.
\end{equation}
It is apparent from the above objective function that the terms involving $z_i^{\star}$ are:
\begin{align*}
g(z^{\star}_{i}) = \alpha_i |z^{\star}_i - \bar{\alpha}_i| + \beta_i |z^{\star}_i - \bar{\beta}_i| + \gamma_i |z^{\star}_i - z^{\star}_{i-1}| + \gamma_{i+1} |z^{\star}_{i+1} - z^{\star}_i|.
\end{align*}
From Theorem.\ref{theorem}, we know that if the optimal solution is unique, then $z_i^{\star} \in \{\bar{\alpha}_i, \bar{\beta}_i, z_{i-1}^{\star}, z_{i+1}^{\star}\}$. Therefore, either $z_i^{\star} = z_{i-1}^{\star}$ or $z_i^{\star} = z_{i+1}^{\star}$ must hold. Without loss of generality, we assume that $z_{i}^{\star} = z_{i-1}^{\star}$. 
Substituting $z^{\star}_i = z^{\star}_{i-1}$ into the objective function \eqref{prove:eq}, which yields:
\begin{equation}
\begin{aligned}
f(\bm{z^{\star}})|_{z^{\star}_i=z^{\star}_{i-1}} &= \sum_{j = 1}^{i-2} \alpha_j| z^{\star}_j-\bar{\alpha}_j| + \alpha_{i-1}| z^{\star}_{i-1}-\bar{\alpha}_{i-1}| + \alpha_i| z^{\star}_{i-1}-\bar{\alpha}_i| + \sum_{j=i+1}^{N} \alpha_{j+1}| z^{\star}_{j}-\bar{\alpha}_{j}|  \\
& + \sum_{j = 1}^{i-2} \beta_j| z^{\star}_j-\bar{\beta}_j| + \beta_{i-1}| z^{\star}_{i-1}-\bar{\beta}_{i-1}| + \beta_i| z^{\star}_{i-1}-\bar{\beta}_i| + \sum_{j=i+1}^{N} \beta_{j}| z^{\star}_{j}-\bar{\beta}_{j+1}|\\
& + \sum_{j=1}^{i-2} \gamma_j|z^{\star}_j-z^{\star}_{j-1}| + \gamma_{i-1}|z^{\star}_{i-1}-z^{\star}_{i-2}| + \gamma_{i+1}|z^{\star}_{i+1}-z^{\star}_{i-1}| + \sum_{j=i+2}^{N+1}\gamma_{j}|z^{\star}_{j}-z^{\star}_{j-1}|.
\end{aligned}
\end{equation}
According to the Theorem.\ref{theorem}, we can further deduce that $z_{i}^{\star}$ follows that:
$$
z_{i}^{\star} \in \{\bar{\alpha}_{i-1}, \bar{\alpha}_{i}, \bar{\beta}_{i-1}, \bar{\beta}_{i}, z_{i-2}^{\star}, z_{i+1}^{\star}\}.
$$
Since \( z_{i}^{\star} = z_{i-1}^{\star} \) and \( z_{i}^{\star} \notin \mathcal{Z} \), it can be deduced that one of two conditions must be satisfied: either \( z_i^{\star} = z_{i-1}^{\star} = z_{i-2}^{\star} \) or \( z_i^{\star} = z_{i-1}^{\star} = z_{i+1}^{\star} \). By recursively applying this logic, we ultimately arrive at the conclusion that \( z_i^{\star} \in \mathcal{Z} \). This finding contradicts our initial assumption that \( z_i^{\star} \notin \mathcal{Z} \), which completes the proof. 
\end{proof}
Indeed, the proof we have presented demonstrates that the optimal solution to the optimization problem \eqref{work1:eq41} is a member of a finite set \( \mathcal{Z} \). However, when \( N \) is large, the task of enumerating all possible combinations within \( \mathcal{Z} \) to identify the optimal solution becomes challenging. This difficulty arises from the exponential increase in potential combinations as \( N \) grows. 

The essence of the optimization problem \eqref{work1:eq41} is an L1-norm optimization problem, which is difficult to solve directly using the gradient descent method due to the discontinuity of its derivatives. We will introduce two numerical methods to solve the objective optimization problem. Based on physical practice, the first method proposed in our previous work \cite{2021The} chooses one reference field to approximate the global optimal reference field. In contrast, the second method utilizes the characteristics of this problem to transform it into a linear programming problem with equality constraints.

\subsection{The reference field method}
Although the reference field method cannot guarantee the search for the global optimal reference field, it can ensure that the reference field found is better than the classical IMC method (with $\widetilde{B} = 0$) and the Monte Carlo method based on the classical reference field (with $\widetilde{B} = B$), and has been proved to achieve better results and higher efficiency in practice. First, we select three reference solutions as follows:
\begin{equation}\label{three:candidate}
\begin{aligned}
&\widetilde{Z}^{(1)}=(\bar{\alpha}_1,\cdots,\bar{\alpha}_N),\\
&\widetilde{Z}^{(2)}=(\bar{\beta}_1,\cdots,\bar{\beta}_N),\\
&\widetilde{Z}^{(3)}=(0, \cdots, 0).
\end{aligned}
\end{equation}
Next, we compute the objective values $f(\widetilde{Z}^{(1)}), f(\widetilde{Z}^{(2)}), f(\widetilde{Z}^{(3)})$, and then sort the results in ascending order, denoting them as $h_1, h_2, h_3$. 
Then we relabel the variables of the corresponding sequence $h_{1}, h_{2}, h_{3}$ as  $\widetilde{Z}^{(1)}, \widetilde{Z}^{(2)},\widetilde{Z}^{(3)}$, respectively. Physically, it is reasonable to suppose that a ``good" solution should be close to $\widetilde{Z}^{(1)}$. Here, a perturbation coefficient is introduced as follows:

\begin{equation}
\delta=\lambda\frac{h_1}{h_1+h_2}.
\end{equation}
Here, $\lambda$ is a constant with a value range of $(0,1]$. From above definition, we redefine the alternative solution $\widetilde{Z}^{(3)}$ as:
\begin{equation}
\widetilde{Z}^{(3)}=(1-\delta) \widetilde{Z}^{(1)}+
\delta\widetilde{Z}^{(2)}.\label{work2:1}
\end{equation}
Here, $\widetilde{Z}^{(3)}$ in the right-hand side is the previous defined by \eqref{three:candidate} while the left-hand side is the redefined solution. After obtaining the new $\widetilde{Z}^{(3)}$, we recalculate $f(\widetilde{Z}^{(3)})$. If $f(\widetilde{Z}^{(3)}) < h_2$, we repeat the former procedure to update $h_1, h_2, \widetilde{Z}^{(1)}, \widetilde{Z}^{(2)}$. If $f(\widetilde{Z}^{(3)}) > h_2$, then we set $\delta = \frac{1}{2}\delta$ and recalculate $\widetilde{Z}^{(3)}$. The specific algorithm of this procedure is shown as follows:

1) Set the iteration counter $k = 0$, the scaling factor $\lambda = 1$, and the initial candidate solutions as:
   \begin{align*}
   & \widetilde{Z}^{(1)} = (\bar{\alpha}_1, \cdots, \bar{\alpha}_N), \\
   & \widetilde{Z}^{(2)} = (\bar{\beta}_1, \cdots, \bar{\beta}_N), \\
   & \widetilde{Z}^{(3)} = (0, \cdots, 0) .
   \end{align*}

2) Increment $k = k + 1$, and compute $f(\widetilde{Z}^{(1)}), f(\widetilde{Z}^{(2)}), f(\widetilde{Z}^{(3)})$. Sort the results in ascending order and denote them as $h_1, h_2, h_3$, and relabel the solutions corresponding to $h_1$ and $h_2$ as $\widetilde{Z}^{(1)}$ and $\widetilde{Z}^{(2)}$, respectively.

3) Calculate the perturbation coefficient $\delta = \lambda \frac{h_1}{h_1 + h_2}$, and then recompute $\widetilde{Z}^{(3)}$ using equation \eqref{work2:1}:
\begin{equation}
   \widetilde{Z}^{(3)} = (1 - \delta ) \widetilde{Z}^{(1)} + \delta \widetilde{Z}^{(2)}.
\end{equation}

4) If $f(\widetilde{Z}^{(3)}) < h_2$, return to step (2). If $f(\widetilde{Z}^{(3)}) \ge h_2$, set $\lambda = \frac{1}{2} \lambda$, increment $k = k + 1$, and return to step (3).

5) If $k > 10$, terminate the iteration. Otherwise, repeat the procedure 2) to 4).

Setting a maximum number of iteration steps for convergence is necessary and crucial. Fewer steps may lead to divergence, while more steps affect the computational efficiency. We set the maximum iteration to 10 in the numerical simulation as an empirical value. Although this method is quite primary, it is straightforward to compute and has physical significance. For standard-form heat radiation transport problems, setting $z = 0$ at each time step corresponds to the third alternative solution; for the classical reference field method, the reference field $z = (\bar{\beta}_1, \cdots, \bar{\beta}_n)$ corresponds to the first alternative solution. Our selection scheme guarantees that the solution is "better" than the three alternatives.

\subsection{The linear programming method}
Although the candidate reference field method gives a fast alternative solution, it may not obtain the exact solution and suffers from the divergence iteration. In this work, we find that the \eqref{work1:eq41} can be reformulated as a linear programming problem by variable substitutions, which can be solved exactly and efficiently.

\subsubsection{The linear programming with equality constraints}
First, we reformulate the \eqref{work1:eq41} from the unconstrained problem to a constrained optimization problem, which is given by:
\begin{equation}\label{work2:eq2}
\begin{aligned}
\min \quad & \|X\|_1\\
\mathrm{s.t.} \quad
& X_i=\alpha_i( z_i-\bar{\alpha}_i),\quad i=1,\cdots,N,\\
& X_{N+i}=\beta_i( z_i-\bar{\beta}_i),\quad i=1,\cdots,N,\\
& X_{2N+i}=\gamma_i(z_i-z_{i-1}),\quad i=1,\cdots,N+1.
\end{aligned}
\end{equation}
Here, $X \in \mathbb{R}^{3N+1}$. Eliminating $z_{i}$ by simple manipulations, the constraint conditions of problem \eqref{work2:eq2} can be reformulated as follows:
\begin{equation}\label{linear:1}
\begin{aligned}
& \frac{\beta_i}{\alpha_i}X_i-X_{N+i}=\beta_i\bar{\beta}_i-\beta_i\bar{\alpha}_i,\quad i=1,\cdots,N, \\
&X_{2N+1}-\gamma_1\left(\frac{X_1}{\alpha_1}+\bar{\alpha}_1-z_0\right) =0,\\
&X_{2N+j}-\gamma_j\left(\frac{X_j}{\alpha_j}+\bar{\alpha}_j-\frac{X_{j-1}}{\alpha_{j-1}}-\bar{\alpha}_{j-1}\right) =0, \quad j=2,\cdots,N,\\
& X_{3N+1}-\gamma_{N+1}\left(z_{N+1}-\frac{X_N}{\alpha_N}-\bar{\alpha}_N\right) =0.  
\end{aligned}
\end{equation}
It is apparent that \eqref{work2:eq2} is equal to \eqref{linear:1}. Furthermore, we denote the above constraints in the matrix form $AX = b$, and then the optimization problem \eqref{linear:1} is simplified as:
\begin{equation}\label{work2:eq3}
\begin{aligned}
\min \quad & \|X\|_1,\\
\mathrm{s.t.} \quad & A X=b,
\end{aligned}
\end{equation}
where $A \in \mathbb{R}^{(2N+1)\times(3N+1)}$ and $b \in \mathbb{R}^{2N+1}$. It is easy to see that $A$ is a block sparse matrix where non-zero elements are distributed near the diagonal of the block matrix, which is written as:
\begin{equation}\label{matrix:A}
    A = \begin{pmatrix}
            A_{1, 1} & A_{1,2} & A_{1,3}  \\
            A_{2, 1} & A_{2,2} & A_{2,3}
        \end{pmatrix}_{(2N+1) \times(3N+1)}.
\end{equation}
Here, each block matrix is defined as:
\begin{equation}\label{matrix:A2}
    \begin{aligned}
        A_{1,1} &= \begin{pmatrix}
            \frac{\beta_{1}}{\alpha_{1}} & & &   \\
             & \frac{\beta_{2}}{\alpha_{2}} & &  \\
             & &\ddots & \\ 
             & & & &\frac{\beta_{N}}{\alpha_{N}}
        \end{pmatrix}_{N \times N}\quad
        A_{1,2} = \begin{pmatrix}
            -1 & & &\\
             & -1 & &  \\
             & &\ddots & \\ 
             & &  &-1
        \end{pmatrix}_{N \times N}\quad
        A_{1, 3} = \begin{pmatrix}
             & &\\
             & \bm{0}&\\ 
             & &\\
        \end{pmatrix}_{N \times (N+1)},\\   
        A_{2,1} &= \begin{pmatrix}
            -\frac{\gamma_{1}}{\alpha_{1}} &  & & & \\
            \frac{\gamma_{2}}{\alpha_{1}} & -\frac{\gamma_{2}}{\alpha_{2}} & & &  \\
            & \ddots & \ddots &  \\
            & & \frac{\gamma_{N}}{\alpha_{N-1}} & -\frac{\gamma_{N}}{\alpha_{N}} & \\
            & & & \frac{\gamma_{N+1}}{\alpha_{N}}&
        \end{pmatrix}_{(N+1)\times N}\quad 
        A_{2,2} = \begin{pmatrix}
             & &\\
             & \bm{0}&\\ 
             & &\\
        \end{pmatrix}_{(N+1) \times N}\quad          
        A_{3,3} = \begin{pmatrix}
            1 & & &\\
             & 1 & &  \\
             & &\ddots & \\ 
             & &  &1
        \end{pmatrix}_{(N+1) \times (N+1)}.\\     
    \end{aligned}
\end{equation}
The right-side vector $b$ can also be defined as follows:
\begin{equation}\label{bb}
    \begin{aligned}
        b_{i}   & =\beta_i(\bar{\beta}_i-\bar{\alpha}_i), \quad i=1,2,\cdots, N,\\
        b_{N+1} & = \gamma_1(\bar{\alpha}_1-z_0),\\
        b_{N+j} & = \gamma_j(\bar{\alpha}_j-\bar{\alpha}_{j-1}),\quad j=2,3,\cdots, N,\\
        b_{2N+1} & = \gamma_{N+1}(z_{N+1}-\bar{\alpha}_N).
    \end{aligned}
\end{equation}

Generally, \eqref{work2:eq3} is referred to as the basis pursuit problem, which has widespread applications in compressed sensing, image processing, machine learning, and other fields. To remove the absolute value, the positive part $X_{i}^{+}$ and the negative part $X_{i}^{-}$ are introduced for $i$-th component of $X$. Using these denotations, we obtain that 
\begin{equation}\label{abs}
    \begin{aligned}
    X_{i} & = X_{i}^{+} - X_{i}^{-},\\
    |X_{i}| & = X_{i}^{+} + X_{i}^{-}.
    \end{aligned}
\end{equation}
By substituting \eqref{abs} into \eqref{work2:eq3}, we obtain a linear programming problem with equality constraints, which is written as:
\begin{equation}
    \begin{aligned}
        \min \quad & f(X^{+}, X^{-}) = \sum_{i=1}^{3N+1} X_{i}^{+}+X_{i}^{-}\\
        \mathrm{s.t.} \quad & \begin{pmatrix}
            A & -A
        \end{pmatrix}\begin{pmatrix}
            X^{+}\\
            X^{-}
        \end{pmatrix} = b,\\
        & X^{+}_{i}, X^{-}_{i} \geq 0, i = 1, 2, \cdots, 3N+1.
    \end{aligned}
    \label{work2:eq5}
\end{equation}
Here $X^{+}, X^{-}$ are the positive part vector and the negative part of the vector $X$, whose components are all non-negative. $A$ and $b$ is same as \eqref{matrix:A2} and \eqref{bb}, respectively. Alternatively, we further rewrite the \eqref{work2:eq5} as the standard linear programming formulation:
\begin{equation}\label{work2:eqq6}
    \begin{aligned}
        \min \quad & f(\Tilde{X}) = \sum_{i}^{6N+2} \Tilde{X}_{i}\\
        \mathrm{s.t.} \quad & \Tilde{A}\Tilde{X} = b,\\
        & \Tilde{X}_{i} \geq 0, \quad i = 1, 2,\cdots,6N+2,
    \end{aligned}
\end{equation}
where 
$$
    \Tilde{A} = \begin{pmatrix}
        A & -A
    \end{pmatrix} \in \mathbb{R}^{(2N+1) \times (6N+2)}, \quad \Tilde{X} = \begin{pmatrix}
        X^{+}\\
        X^{-}
    \end{pmatrix} \in \mathbb{R}^{6N+2},
$$
and $\Tilde{X}_{i}$ is the $i$-th component of $\Tilde{X}$.
The standard linear programming problem can be effectively addressed using the classical simplex or interior-point methods, along with their various modifications. While we will not delve into the specifics of these methodologies in this discussion, we will compare their numerical results in the next section. The simplex method typically involves a higher number of iterations with lower computational costs per step, whereas the interior-point method is characterized by a smaller number of iterations, incurring higher computational expenses, particularly within the context of large-scale systems. Regardless of the approach—simplex or interior-point—securing a ``good" feasible solution or a well-chosen initial point is pivotal for minimizing computational expenses and enhancing overall efficiency, which is explored in-depth in reference \cite{NoceWrig06}.

\subsubsection{The fast linear programming solver}
To further enhance the computational efficiency of linear programming problems, we employ a novel graph-neural-network-based fast solver, MindOpt \cite{mindopt}, which leverages graph neural network methods \cite{unknown2023}. Graph neural networks are neural networks defined based on the geometric structure of graphs, including classic graph convolutional networks and message-passing graph neural networks \cite{Wu_2021}. In recent years, graph neural networks have been extensively used to accelerate the solution of optimization problems \cite{ding2019accelerating,sato2019approximation,chen2023representing}, such as Gasse et al. \cite{gasse2019exact} transform a mixed-integer optimization problem into a bipartite graph and employ a graph neural network to expedite the solution. Chen et al. \cite{chen2023representing} propose transforming a linear programming problem into an equivalent weighted bipartite graph by training the weighted bipartite graph to obtain the feasibility of the original problem and predict an optimal solution, which serves as an initial value for subsequent dual simplex methods or interior point methods. 

According to the literature \cite{klee}, the worst-case computational complexity of the simplex method is $\mathcal{O}(2^{d})$, where $d$ represents the number of vertices. In practical scenarios, grid refinement involves adding vertices to the simplex and increasing the number of iterations, which can be computationally intensive. Consequently, the interior point method is often preferred for large-scale problems. It is also crucial to select an appropriate initial value at the start of the iteration. Based on theoretical insights from \cite{chen2023representing}, MindOpt has been developed as a highly efficient solver for various optimization problems, including linear programming \cite{unknown2023}. MindOpt provides a ``good" initial guess for feasibility and solution, allowing the optimal solution to be achieved within a few iteration steps.

\begin{remark}
Based on the properties mentioned in Section \ref{property}, condition \eqref{unique} is generally satisfied in most practical cases, ensuring that the optimization problem has a unique solution. To enhance the robustness of the algorithm, we will use the reference field method to solve cases where MindOpt fails to find the optimal solution.  
\end{remark}

%% file: grf_result.tex
\section{Numerical results}
In this section, several benchmark tests are conducted to validate the efficiency and accuracy of the global optimal reference field (GORF) method. In the following examples, the units are defined as follows: length in centimeters (cm), mass in grams (g), time in nanoseconds (ns), temperature in kilo electron-volts (keV), and energy in $10^{9}$ Joules (GJ). With these units, the speed of light is 29.98 cm/ns, and the radiation constant \( a \) is 0.01372 \(\text{GJ}/\text{cm}^3\text{-keV}^4\).

\subsection{Example 1}
First, we compared the computational efficiency of the MindOpt solver with the classic linear programming method, such as the simplex method and interior-point method, in solving the equation system \eqref{work2:eq2}. We randomly sampled numbers uniformly distributed between 0 and 1 to determine all constants, including \(\alpha_1, \ldots, \alpha_N\), \(\bar{\alpha}_1, \ldots, \bar{\alpha}_N\), \(\beta_1, \ldots, \beta_N\), \(\bar{\beta}_1, \ldots, \bar{\beta}_N\), \(z_0\), and \(z_{N+1}\). We varied \(N\) to be 500, 1000, 5000, and 9000. We generated three data sets for each value of \(N\) and solved the optimization problem using the Simplex, interior-point, and MindOpt solver. Below are the computational results.
\begin{table}[ht]
\centering
\caption{$\rm{N}=500$}
\begin{tabular}{|c|cc|cc|cc|}
\hline
&\multicolumn{2}{|c|}{first set of data}&\multicolumn{2}{|c|}{second set of data}&\multicolumn{2}{|c|}{third set of data}\\
&CPU time(s) & result &CPU time(s) & result & CPU time(s) & result\\
\hline
Simplex&0.0295 &96.9575&0.0311&75.2883&0.0304&87.2351\\
Interior-point &0.0825&96.9575   &0.0283&75.2883&0.0285&87.2351\\
MindOpt&0.0232&96.9575 &0.0235&75.2883&0.0233&87.2351\\
\hline
\end{tabular}
\label{tab:my_label1}
\end{table}

\begin{table}[ht]
\centering
\caption{$\rm{N}=1000$}
\begin{tabular}{|c|cc|cc|cc|}
\hline
&\multicolumn{2}{|c|}{first set of data}&\multicolumn{2}{|c|}{second set of data}&\multicolumn{2}{|c|}{third set of data}\\
&CPU time(s) & result &CPU time(s) & result & CPU time(s) & result\\
\hline
Simplex&0.1183 &198.4790   &0.1372&190.9449	&0.1028	&186.7912\\
Interior-point&0.1206&198.4790    &0.1161&190.9449	&0.1148	&186.7912\\
MindOpt&0.0310&198.4790    &0.0598&190.9449	&0.0318	&186.7912\\
\hline
\end{tabular}
\label{tab:my_label2}
\end{table}

\begin{table}[ht!]
\centering
\caption{$\rm{N}=5000$}
\begin{tabular}{|c|cc|cc|cc|}
\hline
&\multicolumn{2}{|c|}{first set of data}&\multicolumn{2}{|c|}{second set of data}&\multicolumn{2}{|c|}{third set of data}\\
&CPU time(s) & result &CPU time(s) & result & CPU time(s) & result\\
\hline
Simplex&41.0556 &963.4350  &45.0728&979.0820		&33.1757	&796.0812\\
Interior-point&6.5232&963.4347    &4.7127&1979.0819	&5.0965	&796.0812\\
MindOpt&0.1321&963.4347    &0.1705&979.0819		&0.1706	&796.0812\\
\hline
\end{tabular}
\label{tab:my_label3}
\end{table}

\begin{table}[ht]
\centering
\caption{$\rm{N}=9000$}
\begin{tabular}{|c|cc|cc|cc|}
\hline
&\multicolumn{2}{|c|}{first set of data}&\multicolumn{2}{|c|}{second set of data}&\multicolumn{2}{|c|}{third set of data}\\
&CPU time(s) & result &CPU time(s) & result & CPU time(s) & result\\
\hline
Simplex&350.4180 &1754.4630  &344.9500&1750.4412		&327.2850	&1747.3901\\
Interior-point&17.4137&1754.4630   &18.8084&1750.4363		&18.0382	&1747.3901\\
MindOpt&0.1822&1754.4630   &0.1990&1750.4363		&0.2130	&1747.3901\\
\hline
\end{tabular}
\label{tab:my_label4}
\end{table}
Tables \ref{tab:my_label1}, \ref{tab:my_label2}, \ref{tab:my_label3}, and \ref{tab:my_label4} demonstrate that all three methods successfully identify the optimal solution to the optimization problem. However, the MindOpt solver exhibits a considerably higher computational efficiency than the other two methods. Specifically, as \(N\) increases, the efficiency of the MindOpt solver improves significantly. For \(N = 9000\), the MindOpt solver's computational efficiency is 2 to 3 orders of magnitude greater than that of the other methods.


\subsection{Example 2}
In this example, we consider a one-dimensional grey test problem characterized by sharp spatial gradients in the solution. The setup involves a 1 cm box with a density \(\rho\) of \(1 \, \text{g/cc}\). The initial material temperature \(T_{l0}\) is 1 keV in the region where \(x \leq 0.5\) cm, while the initial material temperature \(T_{r0}\) is set to either \(10^{-3}\) keV or \(10^{-5}\) keV in the region where \(x > 0.5\) cm, respectively. A reflecting boundary is placed at \(x = 0\), and a vacuum boundary is located at \(x = 1\) cm. In this example, the opacity is \(\sigma = 1000 \, \text{cm}^{-1}\), and the specific heat is \(0.05 \, \text{GJ/g/keV}\). We adopt a fixed spatial step size of \(\Delta x = 0.0025\) cm and a fixed time step of \(\Delta t = 0.01\) ns, with the stop time set at 1 ns. Both the classical Implicit Monte Carlo (IMC) method and the GORF method are employed to simulate this case. For the GORF method, 20000 Monte Carlo particles are used, while 400000 particles are used for the classical IMC method. The computational results from the classical IMC method with 4 million Monte Carlo particles are used as the reference solution to facilitate a meaningful comparison of statistical errors.
\begin{figure}[ht!]
\centering
\subfloat[$T_{r0}$ is $10^{-3}$ keV]{\includegraphics[width=0.48\textwidth]
{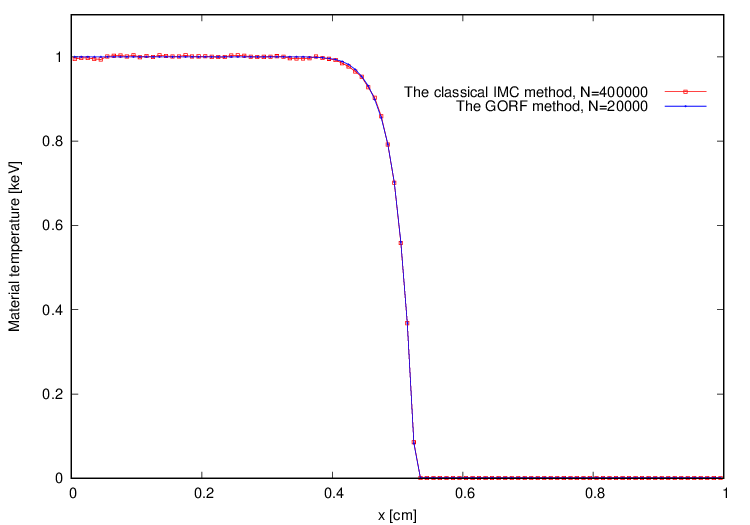}\label{exam3_1a}}
\subfloat[$T_{r0}$ is $10^{-5}$ keV]{\includegraphics[width=0.48\textwidth]{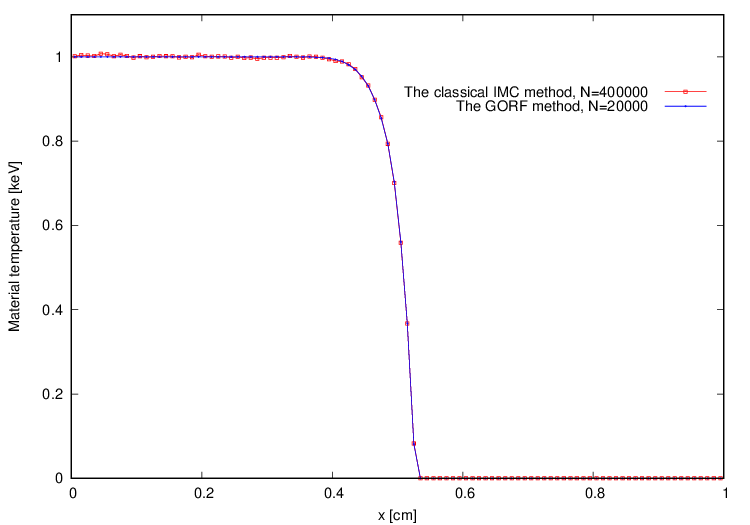}\label{exam3_1b}}
\caption{The numerical results of the spatial distribution of material temperature in different initial material temperatures $T_{r0}$. Both the left and right subfigures show the comparison between the GORF method (blue line) and the classical IMC method (red dot line).}
\label{example3_1}
\end{figure}

Figure \ref{example3_1} illustrates the numerical results for the spatial distribution of material temperature. Figure \ref{exam3_1a} shows the material temperature calculated where $T_{0r}$ is $10^{-3}$ keV. Figure \ref{exam3_1b} shows the material temperature calculated where $T_{0r}$ is $10^{-5}$ keV. In both cases, the computational results of the two methods are basically consistent, validating the correctness of the present numerical methods.
\begin{figure}[ht!]
\centering
\subfloat[$T_{r0}$ is $10^{-3}$ keV]{\includegraphics[width=0.48\textwidth]
{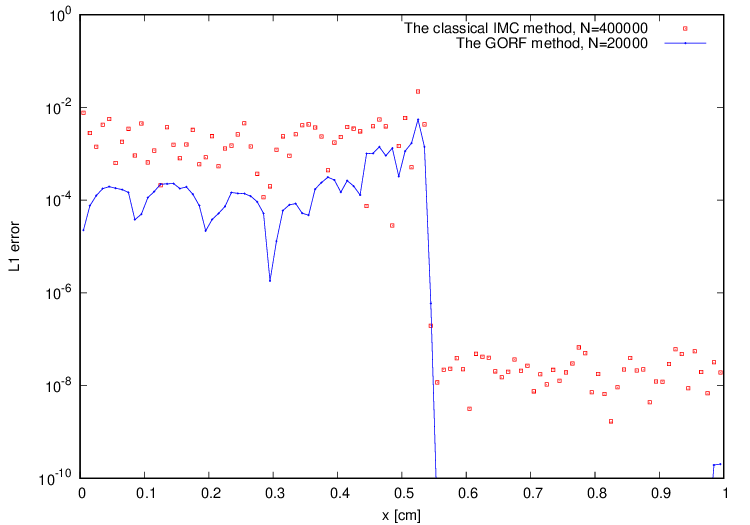}\label{exam3_2a}}
\subfloat[$T_{r0}$ is $10^{-5}$ keV]{\includegraphics[width=0.48\textwidth]{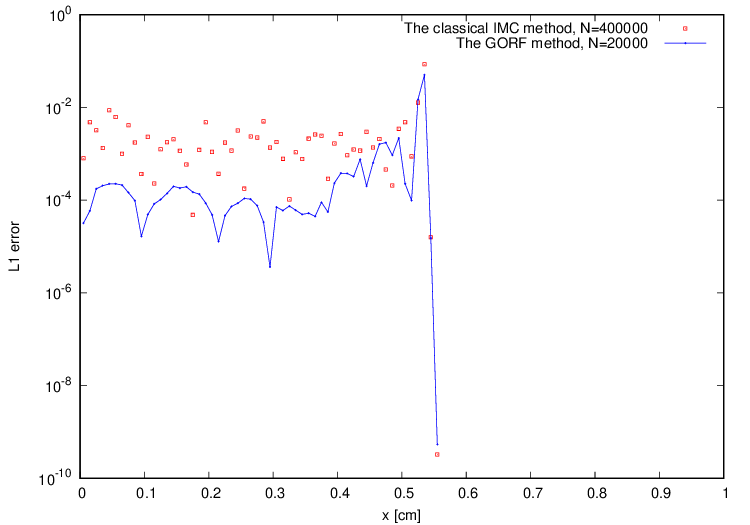}\label{exam3_2b}}
\caption{ Comparison of the relative error between the GORF method (blue line) and the classical IMC method (red dots) in different initial material temperatures $T_{r0}$ in Example 2.}
\label{example3_2}
\end{figure}

Figure \ref{example3_2} compares the relative errors between the two methods. As shown in Figure \ref{exam3_2a}, the GORF method consistently exhibits smaller relative errors than the classical IMC method, both in smooth regions and in regions with discontinuities. In Figure \ref{exam3_2b}, the GORF method demonstrates a significant advantage in smooth regions. However, in areas with large temperature gradients, the statistical errors of the two methods begin to approach each other.

Figure \ref{example3_3} illustrates the CPU time required by both methods to compute up to 1 ns. It is evident that the classical IMC method demands significantly more CPU time due to the need to track 400000 Monte Carlo particles. In contrast, Figure \ref{example3_4} displays the total energy emitted by all radiation sources in each time step for both methods. The GORF method emits considerably less total energy from Monte Carlo particles at each time step compared to the classical IMC method.
\begin{figure}[ht!]
\centering
\subfloat[$T_{r0}$ is $10^{-3}$ keV]{\includegraphics[width=0.48\textwidth]{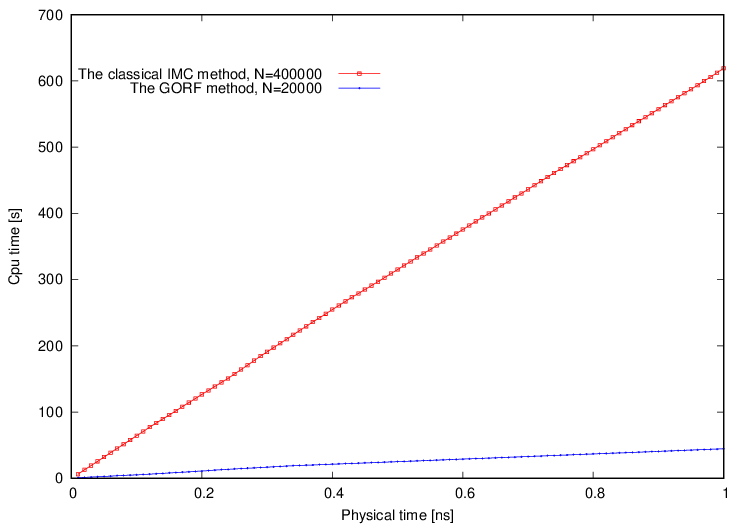}\label{exam3_3a}}
\subfloat[$T_{r0}$ is $10^{-5}$ keV]{\includegraphics[width=0.48\textwidth]{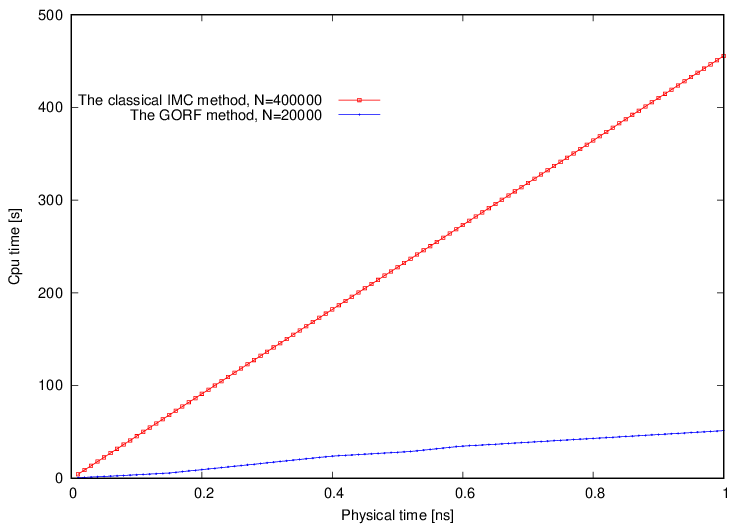}\label{exam3_3b}}
\caption{Comparison of the CPU time between the GORF method (blue line) and the classical IMC method (red dot) in different initial material temperatures $T_{r0}$ in Example 2.}
\label{example3_3}
\end{figure}

\begin{figure}[ht!]
\centering
\subfloat[$T_{r0}$ is $10^{-3}$ keV]{\includegraphics[width=0.48\textwidth]{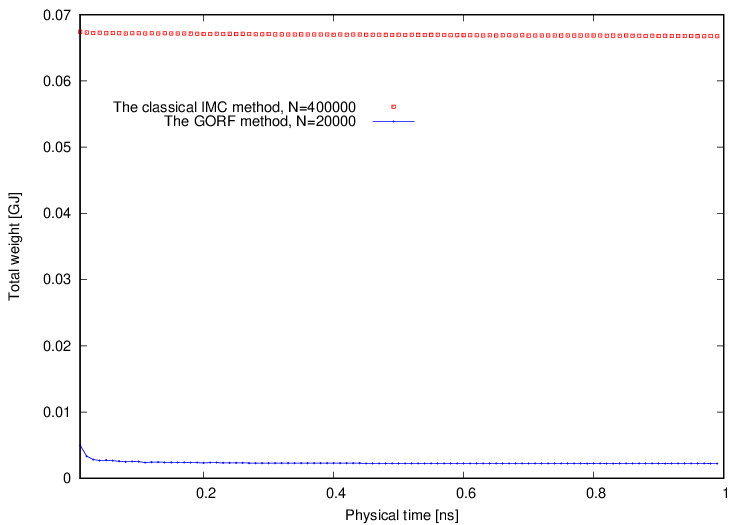}\label{exam3_4a}}
\subfloat[$T_{r0}$ is $10^{-5}$ keV]{\includegraphics[width=0.48\textwidth]{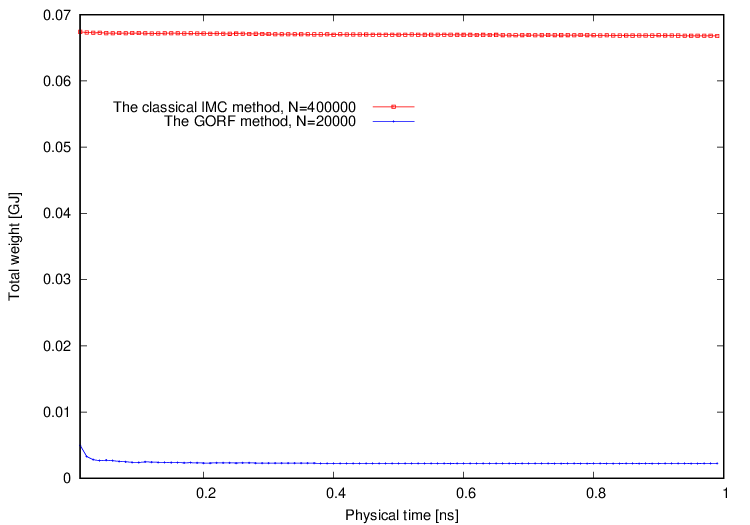}\label{exam3_4b}}
\caption{Comparison of the total energy of Monte Carlo particles emitted per time step between the GORF method (blue) and the classical IMC method (red) in different initial material temperatures $T_{r0}$ in Example 2.}
\label{example3_4}
\end{figure}

In the following numerical examples, we employed two simulation methods. The first method is the GORF method, where the global optimal reference field is computed using the MindOpt solver. The second method is the classical IMC method. For the GORF method, simulations were run with 5000 and 40000 Monte Carlo particles, while the classical IMC method used 320000 particles. To facilitate the comparison of statistical errors, the results from the classical IMC method with 2 million Monte Carlo particles serve as the reference solution. Additionally, we present the relative error in material temperature computed by both methods at each grid cell.

\subsection{Example 3} \label{example3}
The one-dimensional Marshak wave problem under the gray approximation \cite{SUN2015265} is tested in this example. We consider a slab of length 0.4 cm that is initially cold, with a temperature of \(10^{-3}\) keV. The heat capacity is \(c_v = 0.3 \, \mathrm{GJ/cc/keV}\), and the absorption opacity is \(\sigma = 300 \, \mathrm{cm^{-1}}\). A Planck source with a constant temperature of 1 keV is applied at the left boundary, while a reflective boundary condition is imposed at the right boundary. The slab is uniformly divided into 800 cells, with a time step of \(10^{-3}\) ns and the simulation runs until \(t = 15\) ns.
\begin{figure}[ht!]
\centering
\subfloat[The GORF method and classical IMC method]{\includegraphics[width=0.49\textwidth]{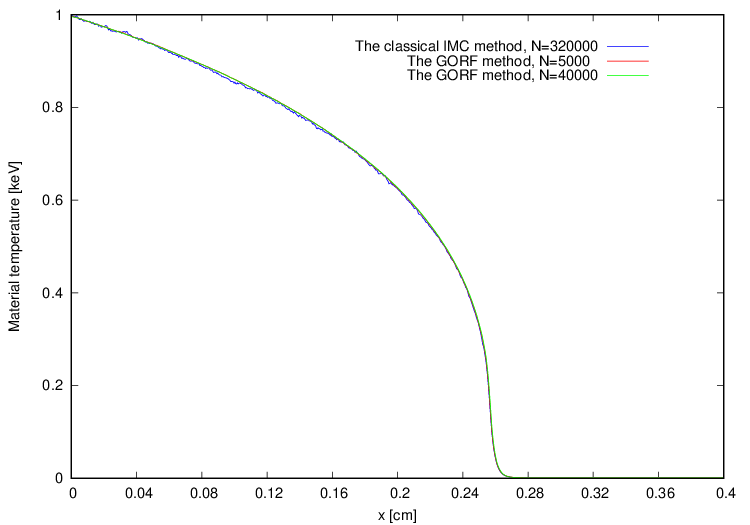}\label{exam1_1a}}
\subfloat[The GORF method and the reference solution]{\includegraphics[width=0.49\textwidth]{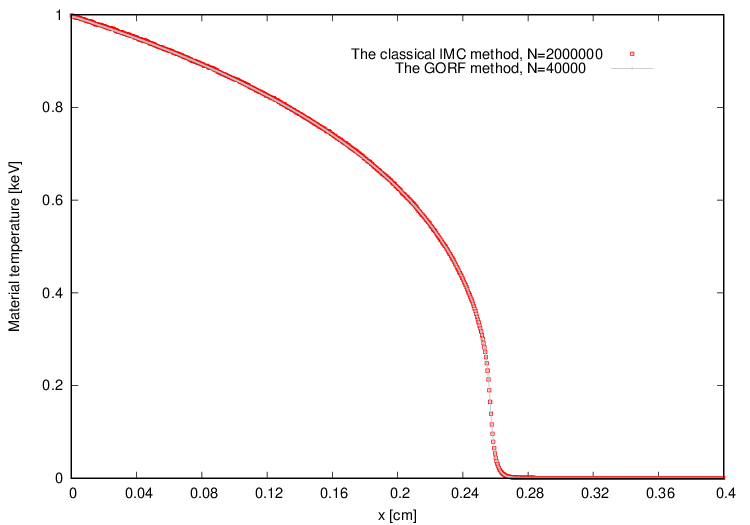}\label{exam1_1b}}
\caption{The comparative numerical results between the GORF method and the classical IMC method on the spatial distribution of material temperature in Example 3. Left (a): spatial distribution of material temperature by the 5000 particles (red line) and 40000 (green line) particles in the GORF method and the 320000 particles in the classical IMC method (blue line). Right (b): spatial distribution of material temperature by the GORF method (grey line) and reference solution (red square).}
\label{example1}
\end{figure}

Figure \ref{example1} presents the numerical results for the spatial distribution of material temperature. In Figure \ref{exam1_1a}, the material temperatures computed by both methods are shown, and it is clear that the results are generally consistent with the reference solution. Figure \ref{exam1_1b} compares the numerical results of the GORF method to the reference solution, demonstrating that the GORF method exhibits significantly lower numerical noise. In contrast, while the statistical noise in the classical IMC method decreases as the number of Monte Carlo particles increases, it remains relatively high overall.

\begin{figure}[ht!]
    \centering
    \includegraphics[width=0.6\linewidth]{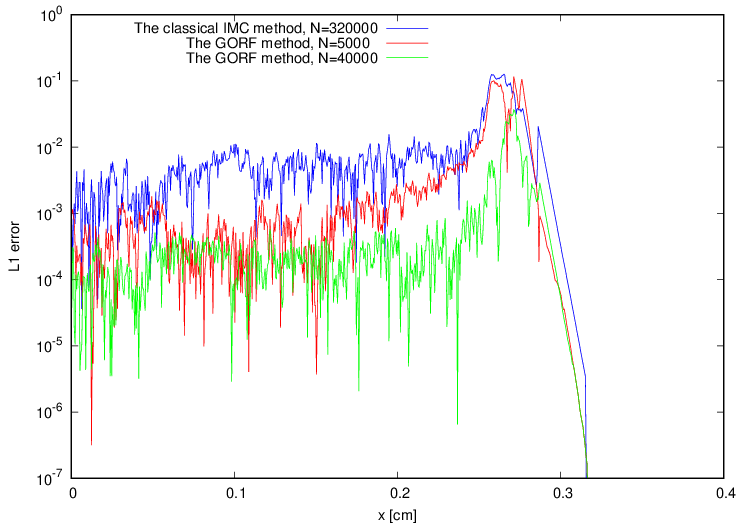}
    \caption{Comparison of the relative error of the GORF method and the classical IMC method in Example 3.}
    \label{exam1_2}
\end{figure}

Figure \ref{exam1_2} compares the relative errors between the two methods. In smooth regions, the GORF method, using only 40000 Monte Carlo particles, achieves relative errors of 1 to 2 orders of magnitude smaller than those of the classical IMC method with 320000 particles. In regions with large temperature gradients, spatial discretization errors become more pronounced, leading to discrepancies between the methods. Nevertheless, the GORF method generally maintains smaller relative errors overall.

\begin{figure}[ht!]
\centering
\includegraphics[width=0.6\textwidth]
{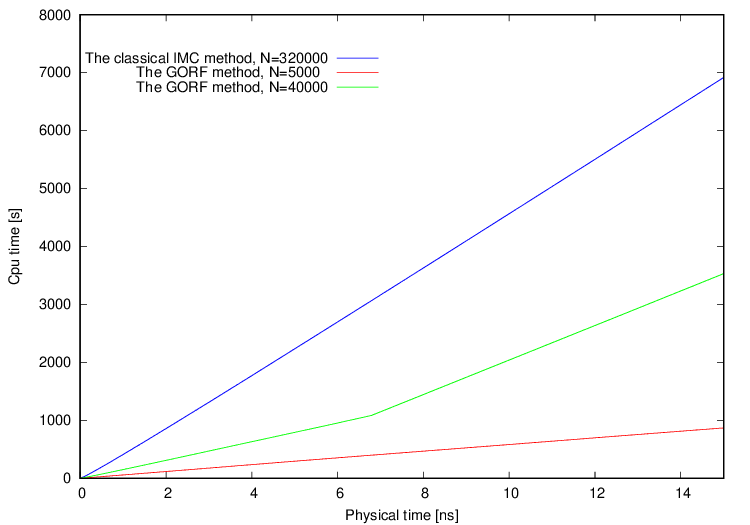}
\caption{Comparison of the CPU time between the GORF method and the classical IMC method in Example 3.}
\label{exam1_3}
\end{figure}

Figure \ref{exam1_3} presents the CPU time consumption for both methods when simulating up to 15 ns. At a Monte Carlo particle count of 40000, the GORF method requires more CPU time than when it is 5000 particles, primarily due to the overhead of solving a linear programming problem. However, as the number of Monte Carlo particles increases to 320000, the classical IMC method demands significantly more CPU time than the GORF method. Furthermore, even at this higher particle count, the relative error of the classical IMC method remains 1 to 2 orders of magnitude larger than that of the GORF method, consistent with earlier observations.


Figure \ref{exam1_4} illustrates the total energy of Monte Carlo particles emitted by all radiation sources at each time step for both methods. As shown, varying the number of Monte Carlo particles has little impact on the GORF method, which remains almost unchanged. In contrast, the GORF method consistently emits significantly less total energy from all radiation sources in each time step compared to the classical IMC method. When combined with the results from Figure \ref{example1} and Figure \ref{exam1_3}, it is evident that the GORF method effectively reduces the total emitted energy, thereby significantly decreasing statistical noise.

\begin{figure}
\centering
\includegraphics[width=0.6\textwidth]
{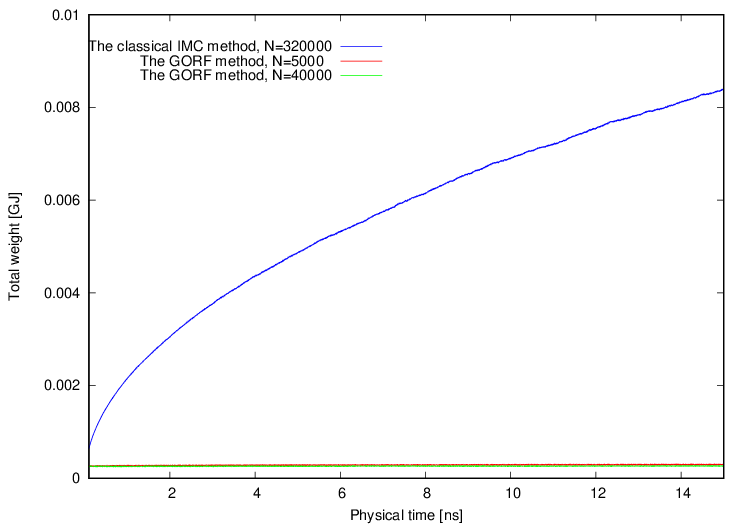}
\caption{Comparison of the total energy of Monte Carlo particles emitted per time step in Example 3.}
\label{exam1_4}
\end{figure}


\begin{figure}[ht!]
\centering
\subfloat[The GORF method and classical IMC method]{\includegraphics[width=0.49\textwidth]{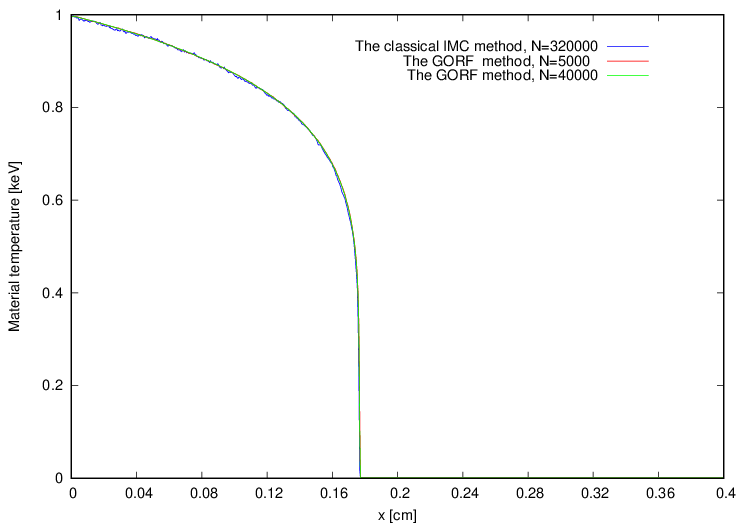}\label{exam2_1a}}
\subfloat[The GORF method and the reference solution]{\includegraphics[width=0.49\textwidth]{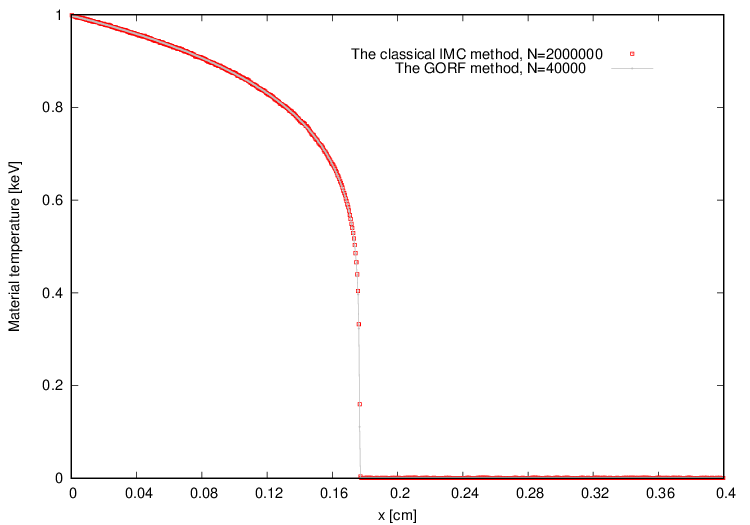}\label{exam2_1b}}
\caption{The comparative numerical results between the GORF method and the classical IMC method on the spatial distribution of material temperature in Example 4. Left (a): spatial distribution of material temperature by the 5000 particles (red line) and 40000 (green line) particles in the GORF method and the 320000 particles in the classical IMC method (blue line). Right (b): spatial distribution of material temperature by the GORF method (grey line) and reference solution (red square).}
\end{figure}
\subsection{Example 4}
In this example, we set the absorption opacity \(\sigma\) to \(300 / T^3 \, \mathrm{cm^{-1}}\) while keeping all other conditions the same as in Example 3, which is originated on \cite{SUN2015265}.

Figures \ref{exam2_1a} and \ref{exam2_1b} display the spatial distribution of material temperature computed by both methods, along with a comparison to the reference solution. Figure \ref{exam2_2} presents a comparison of the relative errors between the two methods, while Figure \ref{exam2_3} illustrates the CPU time required by both methods for computations up to 15 ns. Figure \ref{exam2_4} shows the total energy of Monte Carlo particles emitted by all radiation sources at each time step for both methods. From these figures, we can deduce that the conclusions drawn for Example 4 are similar to those from Example 3. Moreover, these examples confirm that our numerical results align well with theoretical expectations.

\begin{figure}[ht!]
    \centering
    \includegraphics[width=0.6\linewidth]{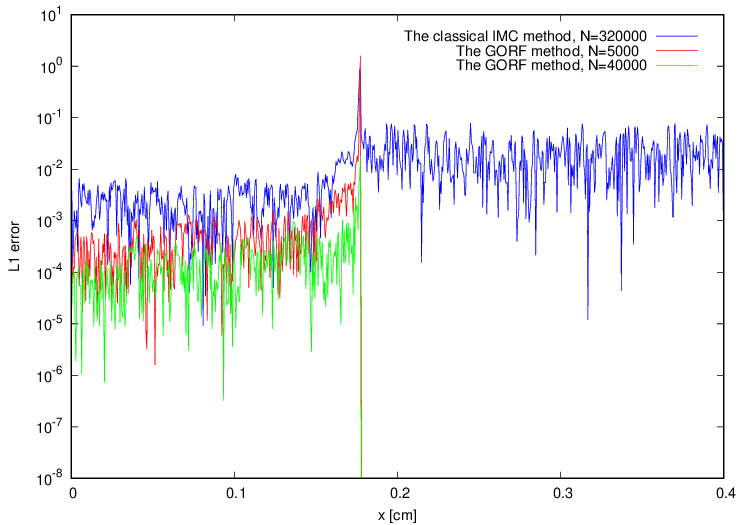}
    \caption{Comparison of the relative error between the GORF method and the classical IMC method in Example 4.}
    \label{exam2_2}
\end{figure}
\begin{figure}[ht!]
\centering
\includegraphics[width=0.6\textwidth]
{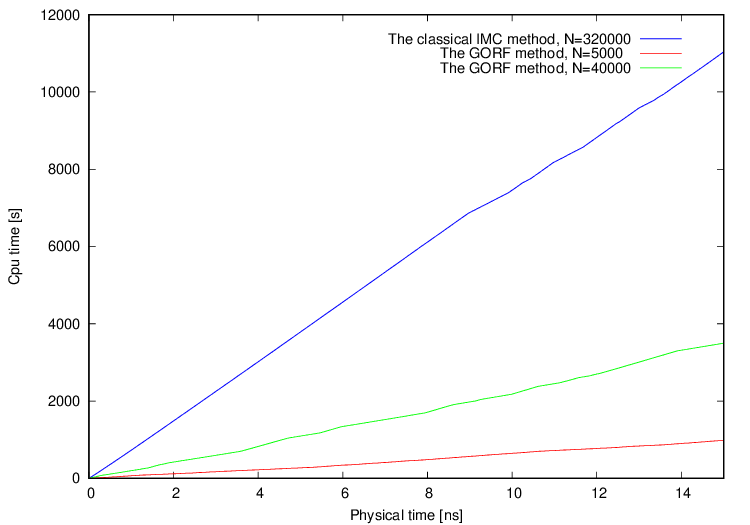}
\caption{Comparison of the CPU time between the GORF method and the classical IMC method in Example 4.}
\label{exam2_3}
\end{figure}
\begin{figure}
\centering
\includegraphics[width=0.6\textwidth]
{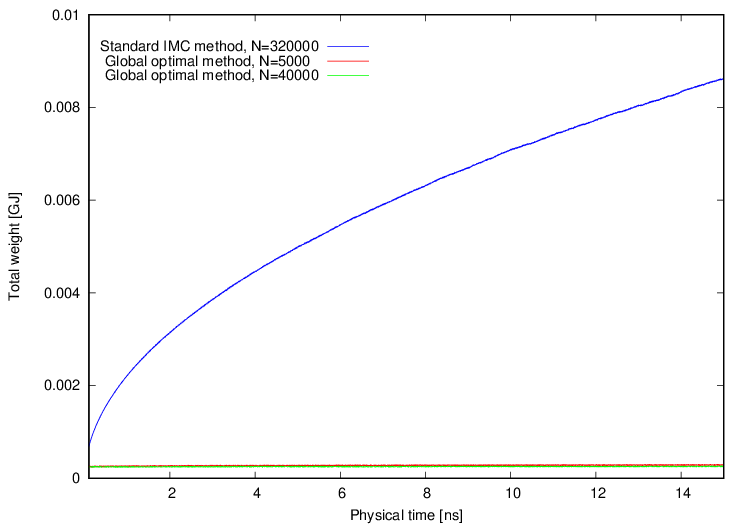}
\caption{Comparison of the total energy of Monte Carlo particles emitted per time step in Example 4.}
\label{exam2_4}
\end{figure}


%% file: grf_conclusion.tex
\section{Conclusion}




This paper introduces a global optimal reference field method based on thermal radiative transport Monte Carlo simulations and derives a generalized reference field form for the radiative transport IMC method. 
Unlike our previous work, which provided an approximate solution to the global reference field, we reduced it to a linear programming problem, which can be solved exactly. Based on this observation, we proposed the global optimal reference field method to reduce the statistic noise with higher performance. 
To compare methods, we used the Simplex method, the interior-point method, and the MindOpt solver. Numerical experiments reveal:

1). The MindOpt solver exhibits superior efficiency in solving linear programming problems, particularly as the problem dimension increases, where its computational advantage becomes more pronounced than the other two methods.

2). In most computational regions, the global optimal reference field method typically shows statistical errors two orders of magnitude lower than those of the classical IMC method.

3). Among the methods compared, the global optimal reference field method using the MindOpt solver demonstrates the highest computational efficiency and a significantly lower performance factor, underscoring its advantage in minimizing statistical errors.

In the future, we will extend our method into the two-dimensional case and explore combining it with traditional variance reduction methods, such as the weight window method, to improve the computational efficiency of Monte Carlo methods for thermal radiation transport problems.

%% file: main.bbl
\begin{thebibliography}{10}

\bibitem{howell1998monte}
John~R Howell.
\newblock The monte carlo method in radiative heat transfer.
\newblock 1998.

\bibitem{wollaber2016four}
Allan~B Wollaber.
\newblock Four decades of implicit {M}onte {C}arlo.
\newblock {\em Journal of Computational and Theoretical Transport}, 45(1-2):1--70, 2016.

\bibitem{WOLLABER200808}
Allan~B. Wollabler.
\newblock {\em Advanced {M}onte {C}arlo methods for themal radiation transport}.
\newblock PhD thesis, University of Michigan, 2008.

\bibitem{noebauer2019monte}
Ulrich~M Noebauer and Stuart~A Sim.
\newblock {M}onte {C}arlo radiative transfer.
\newblock {\em Living Reviews in Computational Astrophysics}, 5:1--103, 2019.

\bibitem{steinberg2022multi}
Elad Steinberg and Shay~I Heizler.
\newblock Multi-frequency implicit semi-analog {M}onte-{C}arlo ({ISMC}) radiative transfer solver in two-dimensions (without teleportation).
\newblock {\em Journal of Computational Physics}, 450:110806, 2022.

\bibitem{mckinley2003comparison}
Michael~Scott McKinley, Eugene~D Brooks~Iii, and Abraham Szoke.
\newblock Comparison of implicit and symbolic implicit {M}onte {C}arlo line transport with frequency weight vector extension.
\newblock {\em Journal of Computational Physics}, 189(1):330--349, 2003.

\bibitem{BROOKS2005737}
Eugene~D. Brooks, Michael~Scott McKinley, Frank Daffin, and Abraham Szöke.
\newblock Symbolic implicit {M}onte {C}arlo radiation transport in the difference formulation: a piecewise constant discretization.
\newblock {\em Journal of Computational Physics}, 205(2):737 -- 754, 2005.

\bibitem{luu2010generalized}
Thomas Luu, Eugene~D Brooks~III, Abraham Szo, et~al.
\newblock Generalized reference fields and source interpolation for the difference formulation of radiation transport.
\newblock {\em Journal of Computational Physics}, 229(5):1626--1642, 2010.

\bibitem{cleveland2010extension}
Mathew~A Cleveland, Nick~A Gentile, and Todd~S Palmer.
\newblock An extension of implicit {M}onte {C}arlo diffusion: Multigroup and the difference formulation.
\newblock {\em Journal of Computational Physics}, 229(16):5707--5723, 2010.

\bibitem{2021The}
Kai Yan.
\newblock The global optimal reference field for the difference formulation in the implicit {M}onte {C}arlo radiation transport.
\newblock {\em Journal of Computational Physics}, 435(1):110258, 2021.

\bibitem{mindopt}
MindOpt.
\newblock Mindopt studio, 2024.

\bibitem{wang2019quantitative}
Dandan Wang and Huaichun Zhou.
\newblock Quantitative evaluation of the computational accuracy for the {M}onte {C}arlo calculation of radiative heat transfer.
\newblock {\em Journal of Quantitative Spectroscopy and Radiative Transfer}, 226:100--114, 2019.

\bibitem{BROOKS1989433}
Eugene~D Brooks.
\newblock Symbolic implicit {M}onte {C}arlo.
\newblock {\em Journal of Computational Physics}, 83(2):433 -- 446, 1989.

\bibitem{FLECK1971313}
J.A. Fleck and J.D. Cummings.
\newblock An implicit {M}onte {C}arlo scheme for calculating time and frequency dependent nonlinear radiation transport.
\newblock {\em Journal of Computational Physics}, 8(3):313 -- 342, 1971.

\bibitem{unknown2023}
Mengyuan Zhang, Wotao Yin, Mengchang Wang, Yangbin Shen, Peng Xiang, You Wu, Liang Zhao, Junqiu Pan, Hu~Jiang, and KuoLing Huang.
\newblock Mindopt tuner: Boost the performance of numerical software by automatic parameter tuning.
\newblock {\em arXiv preprint arXiv:2307.08085}, 2023.

\bibitem{Boyd_Vandenberghe_2004}
Stephen Boyd and Lieven Vandenberghe.
\newblock {\em Convex Optimization}.
\newblock Cambridge University Press, 2004.

\bibitem{NoceWrig06}
Jorge Nocedal and Stephen~J. Wright.
\newblock {\em Numerical Optimization}.
\newblock Springer, New York, NY, USA, 2e edition, 2006.

\bibitem{Wu_2021}
Zonghan Wu, Shirui Pan, Fengwen Chen, Guodong Long, Chengqi Zhang, and Philip~S. Yu.
\newblock A comprehensive survey on graph neural networks.
\newblock {\em IEEE Transactions on Neural Networks and Learning Systems}, 32(1):4–24, January 2021.

\bibitem{ding2019accelerating}
Jian-Ya Ding, Chao Zhang, Lei Shen, Shengyin Li, Bing Wang, Yinghui Xu, and Le~Song.
\newblock Accelerating primal solution findings for mixed integer programs based on solution prediction, 2019.

\bibitem{sato2019approximation}
Ryoma Sato, Makoto Yamada, and Hisashi Kashima.
\newblock Approximation ratios of graph neural networks for combinatorial problems, 2019.

\bibitem{chen2023representing}
Ziang Chen, Jialin Liu, Xinshang Wang, Jianfeng Lu, and Wotao Yin.
\newblock On representing linear programs by graph neural networks, 2023.

\bibitem{gasse2019exact}
Maxime Gasse, Didier Chételat, Nicola Ferroni, Laurent Charlin, and Andrea Lodi.
\newblock Exact combinatorial optimization with graph convolutional neural networks, 2019.

\bibitem{klee}
Victor Klee and George~J. Minty.
\newblock How good is the simplex algorithm?
\newblock In {\em Inequalities, {III} ({P}roc. {T}hird {S}ympos., {U}niv. {C}alifornia, {L}os {A}ngeles, {C}alif., 1969; dedicated to the memory of {T}heodore {S}. {M}otzkin)}, pages 159--175. Academic Press, New York-London, 1972.

\bibitem{SUN2015265}
Wenjun Sun, Song Jiang, and Kun Xu.
\newblock An asymptotic preserving unified gas kinetic scheme for gray radiative transfer equations.
\newblock {\em Journal of Computational Physics}, 285:265 -- 279, 2015.

\end{thebibliography}
